\documentclass[11pt]{article}
\usepackage{graphicx} 

\usepackage{amsmath,amsthm,nicefrac}
\usepackage{amsfonts, amstext}
\usepackage[margin=1in]{geometry}
\usepackage{mathtools}
\usepackage{tikz}
\usepackage{wrapfig}

\usepackage{natbib}

\usepackage{comment}

\usepackage[font={small,it}]{caption}

\usepackage{setspace}
\usepackage[compact]{titlesec}

\usepackage{enumitem}
\usepackage[breaklinks]{hyperref}
\hypersetup{colorlinks=true,%
            citebordercolor={.6 .6 .6},linkbordercolor={.6 .6 .6},%
citecolor=blue!50!black,urlcolor=black,linkcolor=blue!50!black,pagecolor=black}

\usepackage[nameinlink]{cleveref}
\Crefname{algocf}{Algorithm}{Algorithms}
\crefname{algocfline}{line}{lines}
\Crefname{invariant}{Invariant}{Invariants}
\Crefname{claim}{Claim}{Claims}
\Crefname{subclaim}{Subclaim}{Subclaims}

\usepackage{epsfig}
\usepackage{amsthm,amssymb}
\usepackage{amsthm,amssymb}

\usepackage{mathrsfs}
\usepackage{xspace}
\usepackage{soul}
\usepackage{latexsym}
\usepackage{bm}

\usepackage{framed}

\usepackage[dvipsnames]{xcolor}
\definecolor{DarkGray}{rgb}{0.66, 0.66, 0.66}
\definecolor{DarkPowderBlue}{rgb}{0.0, 0.2, 0.6}
\definecolor{fluorescentyellow}{rgb}{0.8, 1.0, 0.0}
\usepackage[ruled,vlined,linesnumbered,algonl,procnumbered]{algorithm2e}
\SetEndCharOfAlgoLine{}
\SetKwComment{Comment}{\footnotesize$\triangleright$\ }{}

\SetCommentSty{mycommfont}
\SetKwProg{Proc}{Procedure}{:}{}
\SetKwFunction{CleanEdges}{CleanHeavyEdges}
\SetKwFunction{BuildGraph}{BuildGraph}
\SetKwFunction{SetCapacitiesAndSolve}{SetCapacitiesAndSolve}
\SetKwData{True}{true}
\SetKwData{False}{false}

\usepackage{thmtools,thm-restate}

\makeatletter
\allowdisplaybreaks
\makeatother

\newcounter{note}[section]
\renewcommand{\thenote}{\thesection.\arabic{note}}

\newcommand{\aaronnote}[1]{\refstepcounter{note}$\ll${\bf Aaron~\thenote:}
  {\sf \color{blue} #1}$\gg$\marginpar{\tiny\bf AR~\thenote}}

\DeclareMathOperator*{\softmax}{smax}
\DeclareMathOperator*{\softmin}{smin}

\newcommand{\aaron}[1]{\aaronnote{#1}}

\sethlcolor{fluorescentyellow}

\newcommand{\initOneLiners}{%
    \setlength{\itemsep}{0pt}
    \setlength{\parsep }{0pt}
    \setlength{\topsep }{0pt}
}

\newtheorem{theorem}{Theorem}[section]
\newtheorem{lemma}[theorem]{Lemma}
\newtheorem{claim}[theorem]{Claim}

\newtheorem{corollary}[theorem]{Corollary}

\theoremstyle{definition}
\newtheorem{defn}[theorem]{Definition}

\theoremstyle{remark}
\newtheorem{remark}[theorem]{Remark}
\theoremstyle{definition}

\makeatletter 
\renewcommand{\theinvariant}{(I\@arabic\c@invariant)}
\makeatother

\newcommand{\eps}{\varepsilon}
\newcommand{\sse}{\subseteq}

\newcommand{\BB}{\mathbb{B}}
\newcommand{\ZZ}{\mathbb{Z}}
\newcommand{\RR}{\mathbb{R}}

\newcommand{\poly}{\operatorname{poly}}

\newcommand{\polylog}{\operatorname{polylog}}
\renewcommand{\emptyset}{\varnothing}

\newcommand{\calI}{\mathcal{I}}

\newcommand{\junk}[1]{}
\newcommand{\eat}[1]{}

\newif\ifhideproofs

\ifhideproofs
\usepackage{environ}   
\NewEnviron{hide}{}

\fi

\usepackage{bbm}

\newcommand{\congest}{$\mathsf{CONGEST}$\xspace}
\newcommand{\local}{$\mathsf{LOCAL}$\xspace}
\DeclareMathOperator{\loc}{\mathbf{local}}
\DeclareMathOperator{\glo}{\mathbf{global}}
\DeclareMathOperator{\ratio}{\mathbf{ratio}}

\definecolor{aarondarkorange}{RGB}{180,80,0}

\newcommand{\smax}{s_{\textrm{max}}}
\newcommand{\smin}{s_{\textrm{min}}}

\newcommand{\jfail}{J_{\textrm{fail}}}

\newcommand{\ampc}{\mathcal{A}_{\textrm{MPC}}}

\newcommand{\header}[1]{\emph{\underline{#1}}}

\newcommand{\MixPC}{\textsf{MixPC}\xspace}

\newcommand{\distYoung}{\texttt{DistrFracLB}\xspace}
\newcommand{\distBFS}{\texttt{DistrBFS}\xspace}

\title{Distributed Load Balancing on Unrelated Machines}
\author{
Aaron Bernstein\thanks{
New York University. Supported by Sloan Fellowship, Google Research Fellowship,  NSF Grant 1942010, and Charles S. Baylis endowment at NYU. 
}
\and
Anupam Gupta\thanks{
Department of Computer Science,
Courant Institute of Mathematical Sciences, New York University.
Supported in part by NSF awards CCF-2422926 and CCF-2608359.
}
\and
Zhaozi Wang$^\dagger$
}

\date{July 2026}

\begin{document}

\maketitle

\thispagestyle{empty}

\begin{abstract}

We study the well-known \emph{load balancing} problem in the
distributed (\congest) model of computation. The input is a bipartite
graph $G = (J \cup M, E)$, with \emph{jobs} $J$ to be assigned to
\emph{machines} $M$. We consider the \emph{unrelated machines}
setting, where the input specifies an arbitrary non-negative size
$s_{ij}$ for every machine $i$ and job $j$. The goal is to find an
assignment $\varphi: J \to M$ that minimizes the maximum
\textit{machine load}, where the load of a machine $i$ is
$\sum_{j: \varphi(j) = i} s_{ij}$, i.e., the total size of the jobs
assigned to it.

\medskip
In the \congest model, the state-of-the-art is an algorithm that runs
in $\polylog$ rounds and returns a $(1+\eps)$-approximate fractional
solution \cite{Ahmadian2021DistributedLB}. This can be combined with a
rounding scheme to also yield a $(2+\eps)$-approximate integral
solution, which is the best polynomial-time result known even in the
centralized setting.
However, this algorithm, 
as well as
all previous \congest algorithms, can only solve a special case of
load balancing, where all edges incident to a job $j$ have the same
size $s_j$.

\medskip Our main contribution is a \congest algorithm that can handle
the case of general sizes $s_{ij}$. The algorithm is again essentially
optimal: it computes a $(1+\eps)$-approximate fractional solution or a
$(2+\eps)$-approximate integral solution in $\polylog$ rounds. 
  The problem structure changes significantly
once we allow arbitrary edge-sizes, so our techniques are very different
from those used in previous algorithms for distributed load balancing.


\medskip One ingredient of our result is a black-box tool that is of
independent interest: a $(1+\eps)$-approximation algorithm to
arbitrary mixed packing-covering linear programs in the \congest model
in $\polylog$ rounds. While such an algorithm was already known in the
more powerful parallel
model, 
previous $\polylog$-round algorithms in the distributed \congest model only solved
\emph{pure} packing or \emph{pure} covering
problems. 
We improve upon a very recent \congest algorithm for mixed
packing-covering that runs in $O(D \, \polylog)$ rounds, where $D$ is the
diameter of the corresponding communication graph.


\end{abstract}

\newpage



\setcounter{page}{1}

\maketitle

\section{Introduction}
\label{sec:introduction}


In this work we study the \emph{load balancing} problem in the
distributed (\congest) model of computation. The input is a bipartite
graph $G = (J \cup M, E)$, where we refer to the sets $J$ and $M$ as
the \emph{jobs} and \emph{machines}, respectively. We consider the
\emph{unrelated machines} model, where the size of each job $j$ on a
machine $i$ is denoted $s_{ij}$, and these sizes can be completely
unrelated to each other. (If there is no edge $(i,j) \in E$, then we
imagine that $s_{ij} = \infty$; i.e., assigning job $j$ to machine $i$
is forbidden.) The goal is to find an assignment $\varphi: J \to M$ that minimizes the maximum \textit{machine load},
where the load of a machine $i$ is $\sum_{j: \varphi(j) = i} s_{ij}$,
i.e., the total size of jobs assigned to it.

There is a rich literature on load balancing, both within the graph
algorithms community (where it is sometimes called \emph{semi-matching}), and
the scheduling community (where it is called \emph{makespan
  minimization}). It has been studied within many different models of
computation, from classical polynomial-time approximation algorithms
\cite{LenstraST90,ShmoysT93}, to fast algorithms \cite{HarveyLLT06,
  FakcharoenpholLN14}, online algorithms (both in the classical
model~\cite{AzarNR95,AspnesAFPW97}, as well as in settings with
recourse~\cite{GKS14-matching,krishnaswamy2023online} and other
beyond-worst-case
scenarios~\cite{argue2022learning,GM26-soda,LX21,im2024online}),
streaming algorithms \cite{KonradR16,AssadiBL23,AssadiBLLW25},
parallel algorithms \cite{Young2001MixedPackingCovering,
  Mahoney2016ParallelMixedPackingCovering, Li23}, and dynamic
algorithms \cite{BhattacharyaKS23}. In the standard centralized model,
one can solve the fractional version exactly via linear programming,
and the best known integral solution is the 2-approximation of
\cite{LenstraST90,ShmoysT93}.
Thus, in most of these models of computation, the gold standard is
typically a $(1+\eps)$ approximate fractional solution or a $(2+\eps)$
approximate integral solution, though these are not always achievable.

\paragraph{Distributed Model.}
All of our distributed algorithms are in the standard \congest model
of distributed computation~\cite{Peleg00}. 
Communication proceeds in synchronous rounds over the underlying communication graph, and in each round every node in the graph can send $O(\log N)$ bits to each neighbor, where $N$ is the number of nodes in the graph. 


\paragraph{Distributed Load Balancing.}
For load balancing, the communication graph is the input graph
$G=(J\cup M,E)$ itself. Initially, each node knows its own type (job
or machine), its incident edges, and the corresponding sizes $s_{ij}$
on those edges. At termination, each machine locally knows the jobs
assigned to it. There are several previous results on this problem in the somewhat narrower \emph{restricted-assignment setting}, where all edge
sizes are in $\{1,\infty\}$ \cite{CzygrinowHSW16, HalldorssonKPR18, AssadiBL20, Ahmadian2021DistributedLB}. The state of the art is the result of  \cite{Ahmadian2021DistributedLB}, 
which computes a
$1+\eps$-approximate fractional solution to makespan minimization (and
more) in $\poly(\log N,1/\eps)$ rounds. (While they do not round this
solution, one can apply the black-box rounding scheme we use to get a
$(2+\eps)$-approximate integral
solution.) 
Much of the previous work (including \cite{Ahmadian2021DistributedLB}) also extends to the setting where each job has a non-negative size $s_j$ and the
edge sizes $s_{ij}$ are in $\{s_j, \infty\}$, but as we discuss later, none of the previous results extend to unrelated edge sizes. We ask:
\begin{quote}
  Question 1: \emph{Does there exist a poly-logarithmic round protocol
    to get a $(2+\eps)$-approximation for the general
    (unrelated-machines) load balancing problem in the \congest
    model?}
\end{quote}

\paragraph{Distributed Mixed Packing-Covering LPs.} A promising
approach to solve fractional load balancing in the \congest model is the
following first step: assume we know the optimum load $L$, and the goal is to find
a feasible fractional solution where each machine load is at most
$L$. This is just a mixed packing-covering linear program (where each
constraint is a packing or a covering constraint):
\begin{align*}
  \textstyle\sum_{j: (i,j) \in E} s_{ij} x_{ij} &\leq L \qquad\qquad
                                                  \forall i \in M \\
  \textstyle\sum_{i: (i,j) \in E} x_{ij} &\geq 1 \qquad\qquad
                                                  \forall j \in J\\
  x_{ij} &\geq 0 \qquad \qquad \forall (i,j) \in E.
\end{align*}
In both the sequential and parallel settings, there exist fast
$(1+\eps)$-approximations for such LPs (e.g.,
\cite{Young2001MixedPackingCovering,Mahoney2016ParallelMixedPackingCovering,Wang2016DiameterReduction}). Unfortunately,
no such fast algorithms are known for the \congest model---or even the
less restrictive \local model. We only know polylog-round \congest
algorithms for the more restricted case of pure packing or pure
covering LPs~\cite{BartalBR04}. To the best of our knowledge, the only
algorithm for mixed packing-covering LPs is one proposed recently,
which requires $\widetilde{\Theta}(D)$ rounds; here $D$ is the
diameter of the graph and can be $\Omega(n)$ in the
worst-case~\cite{Vos2026}. This leads us to ask:
\begin{quote}
  Question 2: \emph{Does there exist a poly-logarithmic round protocol
    to approximately solve mixed packing-covering linear programs in
    the \congest model?}
\end{quote}
We answer both questions in the affirmative.

\subsection{Our Contributions}

\paragraph{First Contribution: Mixed Packing-Covering LPs}
Recall that a mixed packing-covering LP requires us to find $x$ such
that $\{ Px \leq p, Cx \geq c, x \geq 0\}$, where the matrices $P, C$
have sizes $m \times n$. This naturally gives a graph topology with
nodes for each row and column (with each node knowing only the
information for its own row/column); the edges of the graph
then correspond to non-zero entries in $P,C$. We make the standard assumption
that all nodes in the graph know the total number of nodes, as well as
an approximation parameter $\eps$. In this setting, we present the
first efficient algorithm for (fractional) mixed packing-covering LPs
in the \congest model for graphs of arbitrary diameter.

\begin{theorem}[Informal version of \Cref{thm:dist-young}]
  \label{thm:mpc-informal}
  Given a mixed packing-covering LP, there exists a \congest algorithm
  that runs in $\poly((\log mn)/\eps)$ rounds and outputs a non-negative
  fractional solution $x$ that satisfies all the covering constraints
  $Cx \geq c$ exactly and satisfies each packing constraints within a
  multiplicative error of $(1+\eps)$---i.e., $Px \leq (1+\eps)p$.
\end{theorem}

Given that mixed packing-covering LPs capture a wide range of
optimization problems, we feel this can be a general and powerful
subroutine for the \congest model, beyond the immediate application to
load balancing.

\paragraph{Second Contribution: Load Balancing}
We give a $(2+\eps)$-approximation for the integer
version of the unrelated-machines load balancing problem. This is essentially optimal, as there is an unconditional lower bound showing that any better-than-two approximation for this problem requires $\Omega(D + \sqrt{n})$, where the diameter $D$ can be as large as $n$ \cite{AhmadiKO18} \footnote{The lower bound of \cite{AhmadiKO18} is technically for perfect matching. Our stated lower bound follows because if we set all edge-sizes to $1$, then a perfect matching is equivalent to an assignment with maximum load $1$, whereas any other assignment has load at least $2$.}.

At first glance, it may seem
that this result follows directly from the mixed packing-covering result, but
this is not true; writing the LP above required knowing the optimal
maximum load $L$! As we discuss in the technical overview
(\S\ref{sec:overview}), the standard ``guess and check'' methods do
not work in the \congest model---since there is no global
coordination---and several new ideas are needed.

Under the distributed load-balancing conventions stated above, all
nodes know $|J|$ and $|M|$ (or polynomial upper bounds on these
quantities), as well as the approximation parameter $\eps$. However,
we \emph{do not} assume any global knowledge about the maximum or
minimum sizes $s_{ij}$, and our round-complexity \emph{does not} have
any dependence on the aspect ratio
$\log(\textrm{max-size}/\textrm{min-size})$.

\begin{theorem}[Informal version of \Cref{thm:lb-final}]
\label{thm:intro-lb} There is a \congest algorithm that, given any load-balancing
  instance $G = (J \cup M, E)$ with arbitrary positive sizes $s_{ij}$,
  can compute a $(2+\eps)$-approximate integral solution in
  $\poly(\log mn/\eps)$ rounds. Specifically, at termination, every
  machine locally knows which jobs are assigned to it, with the
  guarantee that for every machine $i$,
  \[ \sum_{j \text{ assigned to } i} s_{ij} \leq
  (2+\eps)\lambda^*, \] where $\lambda^*$ is the maximum load of the optimal integral
  solution.
\end{theorem}

\begin{remark}
The same algorithm with a minor tweak can instead return a \emph{fractional} assignment with maximum load $(1+
\eps) \lambda^*_f$, where $\lambda^*_f$ is the maximum load of the optimum fractional assignment. 
\end{remark}




\subsection{Our Techniques}
\label{sec:overview}





\textbf{Mixed Packing-Covering LPs:} Our distributed algorithm for
mixed packing-covering (MixPC) LPs (Theorem \ref{thm:mpc-informal} builds upon the parallel algorithm of
Young~\cite{Young2001MixedPackingCovering}, which does
$O(n\poly((\log n)/\eps))$ work and has $O(\poly((\log n)/\eps))$
depth.  Loosely speaking, Young's algorithm defines, for every column
$j$ of constraint matrix $P$ or $C$, a function $\loc_j$ that depends
only on information local to the column node $j$ (i.e., information
relevant to variable $x_j$); for our purposes, 
$\loc_j$ can be computed in a \emph{single} \congest round.
It also maintains a function $\glo$ the depends on the soft-min (resp.,
soft-max) of the extent to which the current solution $x$ satisfies
the covering (resp., packing) constraints. Define a variable $x_j$ to
be \emph{improving} if:
\[ \loc_j \leq (1+\eps) \glo. \] Young shows that there exists at
least one improving variable, and that the algorithm which repeatedly
finds an improving $x_j$ and slightly increases it, produces a valid
output. He also shows how to speed up the approach by finding
multiple improving $x_j$s in a single parallel round.

The key challenge to porting this algorithm to the distributed setting
is that the function $\glo$ depends on all the variables, so computing
it in the \congest model requires $\Omega(\text{diameter})$ rounds. We
aim to use only a poly-logarithmic number of rounds, so it is
impossible in general for all nodes to agree on even a single bit. (By
contrast, the $ \Omega(\text{diameter})$-round \congest algorithm of
\cite{Vos2026} is able to compute $\glo(x)$ and
broadcast it to all nodes).

To resolve this issue, we show that, somewhat to our own surprise, it
is possible to push the global soft-min/max functions of Young's
algorithm into the analysis; the algorithm itself never needs to
compute them. Because the original function $\glo$ is monotonically
increasing, we show that we can effectively replace $\glo$ with a
single number $\tau$ that increases at a fixed rate (i.e., gets
multiplied by $(1+\eps)$ every poly-logarithmically many rounds) and
can thus be maintained locally by all vertices without any
communication.

\medskip \textbf{Fractional Load Balancing:} As mentioned above,
na\"{\i}vely relaxing the integer LP for load-balancing by dropping
the integrality constraints gives us an MixPC LP (which can be solved
using the above algorithm---let us call it $\ampc$). But this
na\"{\i}ve relaxation can lead to a large integrality gap. To fix this
problem, we need to add an additional constraint: we must require that
only edges in $E_{\lambda^*}$ have non-zero flow, where $E_L$ is
defined as the set of edges $(i,j)$ corresponding to sizes
$s_{ij} \leq L$. If all nodes knew $\lambda^*$, this would again be
easy, since each node could drop the edges with large $s_{ij}$ values
locally, and then run $\ampc$. Even though we do not know $\lambda^*$,
the solution is easy in the centralized (sequential or parallel)
setting: we ``guess-and-check''. We start with a large (but not too
large) guess of the threshold $L$, and if $\ampc$ results in a
feasible solution on the edge set $E_L$, we set $L \gets L/(1+\eps)$
and try again.

But of course, this ``guess-and-check'' approach has a problem in the
\congest model: the nodes cannot agree on whether $\ampc$ with load
threshold $L$ succeeded or failed. Indeed, if $L < \lambda^*$, some
nodes may fail while others successfully assign themselves to a low-load
machine. Since we want a diameter-independent bound, the nodes that
failed and hence know that target load $L$ is infeasible cannot
communicate this information to the rest of the graph in
$o(\text{diameter})$ steps. There is also the problem of having a good
starting guess for $L$; since the nodes only know the job sizes in
their local neighborhood, the number of rounds could be
$O(\log n + \log (s_{\max}/s_{\min}))$, and the second term is
undesirable. To address these issues, we will need a sequence of
ideas.



To keep things simple for now, assume that all nodes know the maximum
and minimum sizes $\smax$ and $\smin$, and also that
$\smax / \smin = \poly(n)$. Since
$\lambda^* \in [\smin, n \cdot \smax]$, each node can start off the
search process with a guess of $L = n\smax$ (for which a solution
trivially exists), and then gradually gradually decrease the target load $L$ geometrically, replacing $L$ by $L/(1+\eps)$ at each step, with all the nodes agreeing on the value of
$L$ as long as the problem is feasible using the subset of edges $E_L$.


\header{Local Failures Cause Headaches.} Let $T$ be the first load for
which some of the nodes fail to satisfy the LP relaxation and let
$\jfail$ be the failed job nodes (we can ensure that machine nodes never fail). Since this is the first failure, the algorithm found a
feasible assignment $x^{T(1+\eps)}$ with maximum load at most $T(1+\eps)$
in the previous phase, using only edges in $E_{T(1+\eps)}$. On the
other hand, the failure at $T$ implies that $\lambda^* > T$, so a
\emph{global} observer would recognize that $x^{T(1+\eps)}$ was an
approximately optimal solution, and that all nodes could just use
$x^{T(1+\eps)}$. The problem is that only $\jfail$ knows
that it is time to stop.
Fixing this locally seems difficult: one approach is for the jobs in
$\jfail$---which know that $T(1+\eps)$ is an acceptable load but $T$
is not---to ``freeze'' their assignments according to
$x^{T(1+\eps)}$. Then the non-failed nodes can continue the algorithm
as before (using the residual capacity not used by the frozen
assignments), trying lower thresholds until they settle on some other
assignment $x^{T'}$ for $T' \leq T$. The problem is that, due to edge deletions, freezing the failed nodes'
assignments could make the resulting problem infeasible for the
non-failed nodes.



\begin{wrapfigure}{R}{0.4\textwidth}
  \centering
  \begin{tikzpicture}[
    vertex/.style={circle, draw, fill=white, inner sep=1.8pt},
    rededge/.style={BrickRed, line width=2.0pt},
    blueedge/.style={blue, line width=1.4pt},
    dots/.style={fill=white, inner sep=1pt}
]

\foreach \i in {1,2} {
    \coordinate (m\i) at (0,{-(\i-1)*1.3});
}
\foreach \i in {1,2,3} {
    \coordinate (j\i) at (4,{-(\i-1)*1.3});
}

\coordinate (mnn) at (0,-4.5);
\coordinate (jnn) at (4,-4.5);

\coordinate (mhiddenBot) at (0,-3.5);

\draw[rededge] (m1) --
    node[above, midway, text=BrickRed, fill=white, inner sep=1pt]
    {$T(1+\epsilon)$}
    (j1);
\draw[rededge] (m2) -- (j2);
\draw[rededge] (mnn) -- (jnn);

\draw[blueedge] (j2) --
    node[below, midway, sloped, text=blue, fill=white, inner sep=1pt]
    {$T$}
    (m1);
\draw[blueedge] (j3) -- (m2);

\draw[blueedge] (jnn) -- (mhiddenBot);

\node[dots] at (2,-3.2) {$\vdots$};

\foreach \i in {1,2} {
    \node[vertex, label=left:{$m_\i$}] at (m\i) {};
}
\node[vertex, label=left:{$m_n$}] at (mnn) {};

\foreach \i in {1,2,3} {
    \node[vertex, label=right:{$j_\i$}] at (j\i) {};
}
\node[vertex, label=right:{$j_n$}] at (jnn) {};

\end{tikzpicture}
  \caption{Example showing problems due to local failures and edge deletions.}
  \label{fig:overview-chain}
\end{wrapfigure}
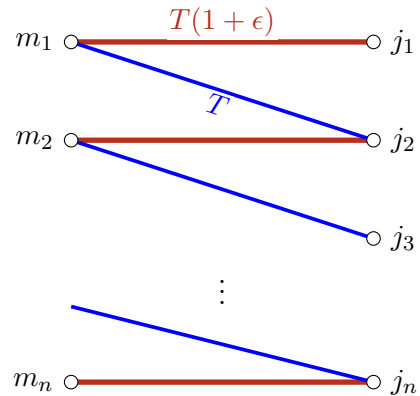

\header{Edge Deletions Cause Further Headaches.} 
As we move from the target load $T(1+\eps)$ to $T$, the algorithm restricts itself to  $E_T$ and deletes edges of size larger than $T$. Consider \Cref{fig:overview-chain},
where the optimal assignment only uses the horizontal red edges of
size $T(1+\eps)$ and hence has $\lambda^* = T(1+\eps)$. In our
strawman algorithm, running $\ampc$ with threshold $L = T(1+\eps)$
succeeds, and for $L = T$ it fails. But this failure is very local:
only job $j_1$ will fail; the other jobs can use the diagonal blue
edges. The failure would cause $j_1$ to freeze its flow according to
$x^{T(1+\eps)}$, i.e., along $(j_1, m_1)$. Now if we re-run $\ampc$
for $L = T$, the node $j_2$ is the only one to fail: it cannot use
$(j_2,m_2)$ since its size $T(1+\eps)$ is too high for $E_{T}$, and it
cannot use edge $(j_2, m_1)$ since $m_1$'s capacity is used up by
$j_1$. Thus, $j_2$ reverts to its flow from $x^{T(1+\eps)}$, causing
$j_3$ to fail next. All in all, we might run $\ampc$ up to $n$ times
to discover all the failed nodes.

\header{A General Barrier for Integer Assignments.}
This example points to a more general barrier: even if $\ampc$ always
returns an \emph{integral} assignment, we cannot hope for our \congest
algorithm to do the same with a $(1+\eps)$-approximation. Indeed, in the example above, the
presence/absence of $j_1$ determines whether the optimal integral
solution is the blue or the red matching. But this means that the
integer assignment of job $j_n$ (along the blue or red edge) depends
on $j_1$'s existence, which would require $\Omega(\text{diameter})$
rounds. 

\header{Our Flow-Interpolation Approach.} We continue to use the idea
of freezing jobs---in fact, we freeze not only the failed jobs, but
also some more jobs. Specifically, the initial set $\jfail$
communicates to all other (non-failed) nodes within a radius of
$O(1/\eps)$, and all those nodes together will define a new frozen
flow based on interpolating between the flow from the previous
iteration, and the one from this iteration. This frozen flow agrees
with the previous iteration's flow near the failed jobs, and slowly
transitions to being the current iteration's flow at the boundary of
the frozen region. For example, when applied to to Figure \ref{fig:overview-chain}, our algorithm would interpolate between the
red and blue matchings: $j_1$ puts flow exclusively on the red edge,
$j_2$ uses mostly red but a bit of blue, $j_3$ uses a bit more blue,
and so on, until $j_{1/\eps}$ only uses its blue edge. (The rest of
the jobs remain unfrozen.) This interpolation incurs an extra $(1+\eps)$-approximation, but
has the advantage that $j_{1/\eps}$ uses only the blue edge
\emph{regardless of the status of $j_1$}; in other words, this
assignment only requires communication within radius $1/\eps$ of the
failed nodes, and hence a small multiplicative overhead in the number
of rounds.

\header{Removing the Assumption on Weights.} In the previous
subsection, we assumed that $\smin,\smax \in [1,\poly(n)]$ and were
globally known. We needed this for two reasons: (a)~all vertices could
agree on an initial load $L = n\smax$ that was guaranteed to be
feasible, and (b)~the algorithm replaces $L$ by $L/(1+\eps)$ at each step, so the round complexity is $O(\log(\smax/\smin)/\eps)$,
and the assumption makes this $\poly((\log n)/\eps)$. We now show how
to use such an algorithm \emph{as a black box} to devise a new
algorithm that no longer requires this assumption.

Our ``shifting'' idea is just as clean as the interpolation idea: we
locally break the instance into instances with known and bounded size
ratio, and solve each one in parallel. Let us sketch the ideas: define
$\beta \approx \log (nm)/\eps$, and define the \emph{class} of an edge
$(i,j)$ to be $\log_{1+\eps}(s_{ij})$.

\begin{enumerate}[nosep]
\item Idea \#1: we can enforce that the smallest and largest edges
  incident to any job differ in class by at most $\beta$: indeed, we
  can drop edges of size more than $(1+\eps)^\beta\; \smin(j)$ and use
  the edge of size $\smin(j)$; this cumulatively increases the machine
  loads by at most a $(1+\eps)$ factor. As a result of this preprocessing, edge classes for each job $j$
  now lie in a \emph{job interval} $I_j$ of width $\beta$; jobs can
  compute this locally.

\item Idea \#2: assume that the nodes agree on a random bucketing of
  the integers into buckets of length $\beta/\eps$, of the form
  \[ \ldots, [a-\beta/\eps, a), \quad [a, a+\beta/\eps), \quad [a+\beta/\eps, a+2\beta/\eps),
    \ldots. \] In expectation, only an $\eps$-fraction of the job
  intervals would cross the boundaries of this random bucketing (or
  even be within $\beta$ distance of the boundaries): call these jobs
  \emph{inactive}.  So each active job interval lies well within some
  bucket---call this the \emph{job bucket}.

\item Idea \#3: Each machine $i$ may get active jobs from various
  buckets. But the load incurred due to all but the highest bucket
  (i.e., the one corresponding to the largest sizes) is insignificant,
  so it can just pre-assign these to itself with very little loss. The
  remaining jobs (active but not pre-assigned) incident to it all have
  the same job bucket---this is the \emph{machine bucket}. Now form
  sub-instances for each machine bucket, and solve them all in parallel. Crucially, each such instance has known $\smax$, $\smin$ with small ratio $\smax/\smin$, and so can be solved using our previous algorithm as a black box.
\end{enumerate}
Finally, choosing the bucketing randomly was only a conceptual trick:
each bucketing is specified by $a \in [\beta/\eps]$, so we enumerate
over all $a$. And since each job is assigned in at least $1-O(\eps)$
fraction of intervals, we finally average all these assignments and
scale them up slightly. 

\header{Rounding to an Integral Solution} The previous subsections
showed how to obtain a fractional solution with maximum load
$(1+\eps)\lambda^*$. The final step is to round this to an integral
solution of maximum load
$(1+\eps)\lambda^* + W^* \leq (2+\eps)\lambda^*$ after reparameterizing $\eps$, where $W^*$ is the
maximum edge-weight used in the fractional solution. We use the
rounding scheme of Shmoys and Tardos~\cite{ShmoysT93}. Specifically,
we can just use the ideas from an efficient implementation of this
scheme by Li~\cite{Li23} in the parallel model, and port these
seamlessly to the \congest model. We leave the details to
\Cref{sec:rounding}, and do not claim any technical novelty here.

\subsection{Related Work}
\label{sec:related-work}

\textbf{Distributed Load Balancing} All prior results on distributed
load balancing assumed job sizes $s_j$; then all edge sizes $s_{ij} \in \{s_j, \infty\}$. 
Czygrinow, Hanckowiak, Szymanska and Wawrzyniak~\cite{CzygrinowHSW16}
showed a $2$-approximation in $O(\textrm{[Max-Degree]}^5)$
rounds. Halld{\'{o}}rsson, K{\"{o}}hler, Patt-Shamir and Rawitz gave
the first algorithm with poly-logarithmic round complexity, but with
the approximation being $O(\log(n)/\log\log(n))$
\cite{HalldorssonKPR18}. Assadi, Bernstein and Langley improved the
approximation to $O(1)$, and also extended the algorithm from the
LOCAL model to the more general CONGEST model
\cite{AssadiBL20}. Finally, Ahmadian, Liu, Peng, and Zadimoghaddam
gives a $1+\eps$-approximation fractional solution in
$O_{\eps}(\polylog)$ rounds \cite{Ahmadian2021DistributedLB}; using a
black-box rounding scheme, this also yields a $(2+\eps)$-approximate integral
solution. Their algorithm also gives a $(1+\eps)$-approximation for
any symmetric and convex function of the machine loads. However, once we allow arbitrary edge sizes, there
are easy instances where no single assignment can be a constant
approximation for both $\ell_1$ and $\ell_\infty$ norms
simultaneously.

The techniques used in these papers are fundamentally limited to
job-dependent sizes, and do not extend to unrelated sizes. 
Indeed, all previous algorithms do one of the following (or both):
\begin{itemize}[nosep]
\item Compute an assignment with no short augmenting paths, or
\item Prove the near-optimality of the computed assignment via
  ``expansion'' arguments. For example, by a simple averaging
  argument, it is easy to see that the optimal load $\lambda^*$
  satisfies $\lambda^* \geq \frac{\sum_{j \in J} s_j}{\|N(J)\|}$ for
  \emph{any} set of jobs $J$ and neighboring machines $N(J)$.
\end{itemize}
By contrast, these properties no longer have a clear interpretation
with arbitrary edge-sizes. For this reason, our algorithm instead
relies on a very different set of techniques that has little overlap
with previous work.


\textbf{Mixed Packing-Covering Linear Programs.} In a very recent
paper, de Vos, Wennmann, Baumecker, Maus and Schager \cite{Vos2026}
study solving mixed packing-covering LPs in the context of the Santa
Claus problem in the \congest model. (In this problem, the goal is to
maximize the minimum machine load, so it is a max-min problem rather
than the min-max problem of load-balancing.) They show that the Santa
Claus problem 
provably requires $\Omega(\sqrt{n} + D)$ rounds, where $D$ is the
diameter of the graph, and hence they can allow their algorithms an
$\Omega(D)$ overhead for all subroutines. This includes solving the
mixed packing-covering LP, for which they can occasionally do global
broadcasts to ensure that all the vertices of the graph are in
sync. By contrast, we aim for $\polylog (mn) \ll D$ rounds, and the
main technical innovations of our paper are precisely to avoid global
broadcasts.



\eat{

\subsection{Text to Use Later}

\aaron{Proposing that we move the rest of the paragraph to the corresponding section of the main body. So ignore it fr now and move on to the next paragraph.} We start with the observation that $\glo$ from Young's algorithm is monotonically increasing. We then define a threshold $\tau$ which is just a single number: all vertices start the same $\tau$ (using their knowledge of $n$), and then all vertices increase $\tau$ at the same fixed rate (i.e. multiply $\tau$ by $(1+\eps)$ every $\polylog$ rounds). We change the requirement for an improving variable $x_j$

$$\loc_j \leq \tau$$

We are able to show that: {\bf 1)} $\tau$ grows quickly enough that the \congest algorithm terminates in $\poly(\log/\eps)$ rounds and
{\bf 2)} $\tau$ grows slowly enough that we always have $\tau \leq \glo$, so our improving edges are also improving by Young's original definition, which guaranties correctness at the end. 

} 


\section{Solving Fractional Load Balancing}
\label{sec:solv-fract-load}

In this section, we prove our result obtaining near-optimal
fractional solutions for distributed load-balancing (Theorem \ref{thm:intro-lb}).


\subsection{Formal Problem Definitions}
\label{sec:LB-problem-defs}

Before we begin, let us formally define instances of the
load-balancing problem and their solutions; our intent is to pin down
their representations in the \congest model.

\begin{defn}[Load Balancing Instance]\label{def:ilb}
  An instance $\calI$ of the \emph{(unrelated-machines) load balancing
    problem} (abbreviated simply as LB) consists of a set $M$ of
  machines and a set $J$ of jobs, where the size of job $j$ on machine
  $i$ is given by a value $s_{ij} \ge 0$; let
  $E := \{(i,j) : s_{ij} < \infty\}$. The goal is to find
  $(x,\lambda)$, with an \emph{integer vector}
  $x = (x_{ij})_{(i,j)\in E}$ that is an optimal solution for the
  integer LP corresponding to the following linear program:
  \begin{align}
    \min \quad & \lambda \notag\\
    \text{s.t.}\quad
      & \textstyle\sum_{j:\,(i,j)\in E} s_{ij}\,x_{ij} \le \lambda
        && \forall i \in M \tag{LB-LP} \label{eq:flb}\\
      & \textstyle\sum_{i:\,(i,j)\in E} x_{ij} \ge 1
        && \forall j \in J \notag\\
      & x_{ij} \in [0,1]
        && \forall (i,j) \in E. \notag
  \end{align}
  We write $\lambda^*(\calI)$ for the optimal value of the
  corresponding integer LP,
  i.e., where each
  $x_{ij} \in \{0,1\}$.
\end{defn}

\begin{defn}[Distributed Load Balancing Instance]\label{def:dilb}
  A distributed load-balancing (LB) instance $\calI$ has an associated
  communication network $H(\calI)$ with a node for each machine $i$
  and one for each job $j$, and an edge $(i,j)$ if $s_{ij} <
  \infty$. (I.e., the edges are precisely the set $E$.) Hence, 
each job $j$ and each machine $i$ know only the sizes
 corresponding to the edges incident to them.
\end{defn}

\begin{defn}[Distributed $(1+O(\eps))$-Approximation for Fractional Load Balancing]
  \label{def:frac-LB}
  A distributed protocol for Fractional LB is an
  $(1+O(\eps))$-approximation if for each distributed instance $\calI$ of
  LB as in \Cref{def:dilb} (where every node is also given the parameter
  $\eps > 0$), it terminates in finitely many rounds, with each
  machine node $i$ returning values $x_{ij} \geq 0$ for all $j \in N(i)$
  such that
  \[
    \sum_{j \in N(i)} x_{ij}s_{ij}
    \leq (1+O(\eps))\lambda^*(\calI),
  \]
  along with $\sum_{i \in N(j)} x_{ij} \geq 1$, where $\lambda^*(\calI)$ is the value of the optimal solution to the integral
  load balancing problem corresponding to (\ref{eq:flb}) on $G$.
  The solution must satisfy one additional property: for every machine $i$ and every $j \in N(i)$, we have
  $x_{ij}=0$ whenever $s_{ij}>(1+\eps)\lambda^*(\calI)$. (We need this property to ensure a small integrality gap.)
  
\end{defn}

Finally, let us define a natural but important distributed representation of sets.

\begin{defn}[Locally represented sets]
  \label{def:local-sets}
  We say a vertex set $S\subseteq M\cup J$ (such as $A_M$, $A_J$, or
  $J_0$ in \Cref{alg:weak-log}) is \emph{locally
  represented} if membership $v\in S$ is given by a single bit stored at node
  $v$, and an edge set $F\subseteq E$ is \emph{locally
  represented} if there is a predicate, evaluable by both endpoints
  from their local state, that determines whether $(i,j)\in F$
  (e.g., $E_L=\{(i,j) \mid s_{ij}\le L\}$, which both endpoints can decide
  since they share $L$).
  Passing such a set to a subroutine requires only that each node knows
  its own membership bit, and requires only $1$ round to make the bits
  consistent across each edge.
\end{defn}


\subsection{Two Black Boxes}
\label{sec:black-box}

We will assume the existence of the following distributed algorithm for a feasibility version of load-balancing, where, given a load-balancing instance and per-machine targets $L_i$, we want to find $x$ such that
\begin{gather}
    \textstyle \sum_{i \in N(j)} x_{ij} \geq 1 \; \forall j, \;\; \sum_{j \in N(i)} x_{ij}s_{ij}
    \leq L_i \; \forall i, \;\; x \geq 0. 
\end{gather} 
    The instance is feasible if and only if such an $x$ exists. 
\begin{lemma}[\distYoung]
\label{lem:mixpc-black-box}
    There is a distributed algorithm which, given a distributed load-balancing instance $\calI$, where each machine is also given a target load $L_i$ and a parameter $\eps$, terminates in $\poly(\log(mn)/\eps)$ rounds. 
   \begin{itemize}[nosep]
       \item Either, upon termination, every machine node returns values $x_{ij}$ for $j \in N(i)$, such that  $\sum_{i \in N(j)} x_{ij} \geq 1$ for all jobs $j$, and 
    $
    \sum_{j \in N(i)} x_{ij}s_{ij}
    \leq (1+O(\eps))\,L_i
    $ for each machine $i$.
    \item 
    Or, the instance $\calI$ is
    infeasible. In this case, at least one of the job nodes returns
    \textsc{reject}. Moreover, the machine nodes never return \textsc{reject};
    and therefore for each machine $i$ we have:
    $
    \sum_{j \in N(i)} x_{ij}s_{ij}
    \leq (1+O(\eps))\,L_i.
    $
   \end{itemize}
\end{lemma}
This result follows from our general result on solving mixed packing-covering problems in distributed settings, which we will discuss in \Cref{sec:distr-mixed-pack}. 


We also assume the existence of following standard distributed BFS process (see e.g. \cite{Peleg2000}.)
\begin{restatable}[\distBFS]{lemma}{distBFSLemma}
  \label{lem:distributed-bfs}
  Let $A_M, A_J, J_0$ be locally represented 
  with $J_0\subseteq A_J$. There is a deterministic CONGEST algorithm
  $\distBFS(A_M, A_J, J_0, \eps)$ running in
  $O(1/\eps)$ rounds, at the end of which every job node $j$
  holds a label $d_j \in \{0,1,\ldots,1/\eps\}\cup\{+\infty\}$
  defined as follows: $d_j = r$ if the shortest path from $J_0$ to $j$
  in the subgraph induced by $A_M\cup A_J$ has exactly $2r$ edges and
  $r\le 1/\eps$, and $d_j=+\infty$ otherwise.
\end{restatable}
\subsection{Algorithm with Known Maximum and Minimum Sizes}
\label{sec:weak-log}

We begin with an algorithm that achieves \Cref{thm:weak-log}; i.e., we
assume that all nodes know upper and lower bounds $s_{\max}$ and
$s_{\min}$ on the sizes. 

\begin{restatable}[Fractional Load Balancing; Weakly Polynomial]{theorem}{weaklog}
  \label{thm:weak-log}
  There is a deterministic distributed $(1+O(\eps))$-approximation algorithm for the Fractional
  Load Balancing problem (Defn.~\ref{def:frac-LB}), where
  additionally, every node in the communication network is given two
  values $s_{\max}$ and $s_{\min}$, which are respectively upper and
  lower bounds on all edge weights.  It works in
  $\poly(\log(mn\cdot s_{\max}/s_{\min})/\eps)$ rounds in
  the \congest model, and treats \Cref{lem:mixpc-black-box} and \Cref{lem:distributed-bfs} as black
  boxes.
\end{restatable}

\paragraph{Algorithm Overview.}

As mentioned before, a central difficulty in the distributed setting is
that we cannot perform a global search on the target load
$\lambda^*$. Instead, the algorithm maintains a common candidate value
$L$, with each node storing and updating its local copy of this
value. Recall from \Cref{sec:overview} that, for any threshold $L$,
$E_L := \{(i,j)\in E : s_{ij}\le L\}$ denotes the set of edges
available at threshold $L$. The search proceeds geometrically: $L$ is initialized at a
trivial upper bound $ns_{\max}$, and in each search iteration, all
nodes synchronously replace $L$ by $L/(1+\eps)$ and delete all
edges with sizes in the range $(L/(1+\eps), L]$, until $L$
reaches $s_{\min}$. During this search, some node may locally report
infeasibility. In this case, the algorithm will fix the $\hat{x}_{ij}$ value of a carefully chosen set of variables ``close'' to the
infeasible node---we say that we ``freeze'' the corresponding edges $(i, j)$, so that these values will no
longer be changed in the future.

In addition to freezing edges, the algorithm will also freeze certain
jobs and machines. The freezing steps are designed to maintain some
invariants:
\begin{enumerate}[nosep]
\item Every frozen job $j$ is already assigned to an extent of
  \[
    \sum_{i \in N(j)} x_{ij} \geq 1.
  \]
\item Every freshly frozen machine $i$ has load $\sum_{j \in N(i)} x_{ij} s_{ij} \leq (1 + O(\eps))\,L$, where $L$ is the
  current step target load. 
\item The remaining unfrozen graph is a valid residual instance for
  target load $L$. More precisely, if $w_i$ is the load placed by the
  frozen edges on machine $i$, then the residual capacity of machine
  $i$ is $L - w_i$. 
  We guarantee that the remaining unfrozen jobs still admit a $(1 +
  \eps)$-feasible fractional assignment with these residual
  capacities; note that this residual assignment is restricted to
  edges $E_L$. Thus each freezing step makes local progress while
  preserving the feasibility needed for the subsequent search.
\end{enumerate}

Without loss of generality, we assume $1/\eps$ is an
integer. If not, we can simply round $\eps$ down to the nearest
unit fraction, which only affects the hidden constants in our
$O(\eps)$ bounds.

\subsubsection{The Algorithm and Description}
\label{sec:weak-algo-desc}

We now present the pseudocode for our algorithm, followed by a description in words (which will make the pseudocode significantly easier to parse). 

\begin{algorithm}
  \caption{\texttt{WeakFractLoadBalancer} (with known size scale)}
  \label{alg:weak-log}
  \KwIn{Each node knows $\eps, s_{\max}, s_{\min}$. Each
    machine node $i$ knows $\forall j \in N(i), s_{ij}$. \\
    $\qquad\qquad$  Each job node $j$ knows $\forall i \in N(j), s_{ij}.$}
  \KwOut{Each machine $i$ outputs values $\hat{x}_{ij}$ for all $j \in N(i)$
  satisfying \Cref{thm:weak-log}.}
  \lForEach(\tcp*[f]{$A_M$ locally represented\ via \Cref{def:local-sets}}){machine $i \in M$}{add $i$ to $A_M$} \label{line:init_A_M} 
  \lForEach(\tcp*[f]{$A_J$ locally represented\ via \Cref{def:local-sets}}){job $j \in J$}{add $j$ to $A_J$} \label{line:init_A_J} 
  \lForEach{machine $i$}{initialize target load $L \gets n\,s_{\max}$, $w_i \gets 0$}
  \lForEach{$i \in M, j \in N(i)$}{$\hat{x}_{ij} \gets 0$, $x^*_{ij}\gets 0$, $x'_{ij}\gets 0$}
  $T \gets \lceil\log_{1+\eps} (n\, s_{\max}/s_{\min})\rceil$
  
  \For{$t \gets 1$ \KwTo $T$}{
    \ForEach{machine $i$}{
    \lIf{$w_i \geq L$ or all jobs $j \in N(i)$ are frozen}{$A_M \gets A_M - \{i\}$\label{line:freeze-machine}}
    }
    \lForEach{machine $i$ and job $j \in N(i)$}{$x'_{ij} \gets x^*_{ij}$}
 $x^* \gets \distYoung(A_M\cup A_J,E_L\cap(A_M\times A_J))$ on 
 following Load Balancing problem\label{line:mpc}:
\begin{align*}
  \textstyle \sum_{j\in A_J:\,(i,j)\in E_L} s_{ij}x_{ij}
  &\le L-w_i
  \quad &&\forall i\in A_M,\\
  \textstyle \sum_{i\in A_M:\,(i,j)\in E_L} x_{ij}
  &\ge 1
  \quad &&\forall j\in A_J,\\
  x_{ij}
  &\ge 0
  \quad &&\forall i\in A_M,  j\in A_J,  (i,j)\in E_L.
\end{align*}
    $J_0 \gets \{\text{ set of job nodes $j$ that return \textsc{reject} } \}$ 
    \tcp*{$J_0$ locally represented\ via \Cref{def:local-sets}}
    $\distBFS(A_M, A_J,  J_0, \eps)$ \tcp*{Every job $j$ gets a label
      $d_j\in \{0, 1, \ldots, 1/\eps\}\cup\{+\infty\}$}
    \ForEach{job $j$}{
      \If{$d_j \neq \infty$}{
        set $\hat{x}_{ij} \gets 
        \eps d_jx^*_{ij} + (1 - \eps\,d_j) x'_{ij}$ and $w_i \gets w_i + \hat{x}_{ij}\cdot
        s_{ij}$ for each $i \in N(j)$ \label{line:edge_setting}\;
        $A_J \gets A_J - \{j\}$ \tcp*{Freeze job $j$ with flow
          $\hat{x}_{\cdot, j}$} \label{line:set_setting}
      }
    }
     \lForEach{node $i$}{update $L \gets L /(1 + \eps)$}
  }
  \lForEach(\tcp*[f]{assign the remaining jobs}){unfrozen job $j \in
    A_J$}{$\hat{x}_{ij} \gets x^*_{ij}$ for each $i \in
    N(j)$} \label{line:unfrozen_setting} 
  \lForEach{machine $i$}{output $\hat{x}_{ij}$ for all $j \in N(i)$}
\end{algorithm}


Before we describe \Cref{alg:weak-log}, here is some notation
used in it:
\begin{enumerate}[nosep,label=(\roman*)]
\item $L$ is the current target load the algorithm is attempting to
  achieve. $\hat{x}$ is the output solution.  $w_i$ is the load
  already given to machine $i$ by previous frozen edges.

\item $x^*$ stores the fractional  solution for the current target load $L$,
  $x'$ stores the fractional solution for the previous target load
  $L(1 + \eps)$. (One exception is when $L = n\,s_{\max}$, in
  which case $x' = 0$.)
\end{enumerate}

We now describe \Cref{alg:weak-log} in words. The sets $A_M$ and $A_J$
are locally represented \emph{active} sets. A machine is said to be active if it is in $A_M$, and frozen precisely when
it is removed from $A_M$; a job is frozen precisely when it is
removed from $A_J$. All jobs and machines are initially active. When any node is frozen (i.e. becomes inactive), it no longer participates in subsequent calls to \distYoung; moreover, when a job $j$ is frozen, its assignment $\hat{x}_{\cdot,j}$ is fixed permanently
\footnote{When a machine $i$ is frozen, its assignment can still change in the subsequent round (Line \ref{line:edge_setting}), but will be permanently fixed thereafter.}.
For
each machine $i$, the value $w_i$ records the load already incurred
due to frozen jobs. Throughout the algorithm we maintain
$w_i = \sum_{j\notin A_J} s_{ij}\hat{x}_{ij}.$ At the beginning of each iteration, a machine $i$ is frozen, i.e.,
removed from $A_M$, if $w_i\ge L$ or all jobs $j\in N(i)$ are frozen.

\paragraph{Inside a Search Iteration:}
At the beginning of an iteration, the algorithm first stores the
previous solution by setting $x' \gets x^*$.  It then runs
\distYoung (Lemma \ref{lem:mixpc-black-box}) on the active instance induced by $A_M\cup
A_J$, restricted to the current edge set $E_L$, with residual capacity
$L-w_i$ on every active machine $i$.  The returned solution is stored
as $x^*$, which is the solution for the current target load $L$.  We
interpret $x^*_{ij}=0$ and $x'_{ij}=0$ for edges that are not present
in the corresponding active instance.

By the guarantee of \distYoung, the returned vector $x^*$ satisfies the
residual packing constraints regardless of the covering-side outputs:
for every active machine $i\in A_M$,
\[
  w_i+\sum_{j\in N(i)\cap A_J} x^*_{ij}s_{ij}
  \le (1+O(\eps))L.
\]
This remains true even if some active job nodes return
\textsc{reject}. If some job nodes return \textsc{reject}, the
algorithm collects them in the set $J_0$ and calls
$\distBFS(A_M,A_J,J_0,\eps)$.  By \Cref{lem:distributed-bfs},
this call runs a bounded BFS from
$J_0$ in the active subgraph $G[A_M\cup A_J]$ and assigns each active
job $j$ a label
\[
  d_j \in \{0,1,\ldots,1/\eps\}\cup\{+\infty\}.
\]
Here $d_j=r$ means that the shortest path from $J_0$ to $j$ has length
$2r$ in the bipartite graph, or equivalently that $j$ is at job-to-job
distance $r$ from $J_0$.

Every job $j$ with finite label is frozen in
this iteration and its final assignment is defined as a convex
combination of the current solution $x^*$ and the previous solution
$x'$:
\[
  \hat{x}_{ij}
  \gets
  \eps d_j x^*_{ij}
  +
  (1-\eps d_j)x'_{ij}
  \qquad \text{for all } i\in N(j).
\]
Intuitively, the jobs closer to the rejecting set use more of the
previous solution $x'$, whereas jobs farther away use more of the
current solution $x^*$, so that the boundary nodes are only using the
current solution and hence interface well with the rest of the graph.
The purpose of this layered convex combination is to smoothly
transition between the solutions $x^*$ and $x'$: the coefficients for
nodes in neighboring BFS layers differ by only $\eps$, which adds only
an $O(\eps)$ fraction additional load on each machine.

After freezing a job $j$, the algorithm adds its contribution
$s_{ij}\hat{x}_{ij}$ to $w_i$ for every neighboring machine $i$, and
then removes $j$ from $A_J$. Note that this update is over all
neighbors $i\in N(j)$, not only over active machines. Thus a machine
that has just been removed from $A_M$ may still receive load in this
freezing step, but only through the previous-solution part $x'$ of the
interpolation. Indeed, if $i\notin A_M$ in the current residual
instance, then $x^*_{ij}=0 \qquad \forall j\in A_J$, and hence any
newly frozen adjacent job contributes
$\hat{x}_{ij} = \eps d_jx^*_{ij}+(1-\eps d_j)x'_{ij} \le x'_{ij}.$
However, this extra load can be charged to the previous iteration's
fractional solution, and hence does not increase the load on the
machine.  After all jobs with finite BFS label have been processed,
the algorithm moves to the next target load $L \gets L/(1 + \eps)$.

\paragraph{After the Search Iterations:} After iterating from
$n s_{\max}$ to $s_{\min}$, the solution $\hat{x}$ contains the
assignments of all frozen jobs. We then use the $x^*$ solution from
the final search iteration to assign the remaining jobs.  In
\Cref{cor:almost_satisfy} we show that after this step, every job receives at least one
unit of mass:
$ \sum_{i\in N(j)} \hat{x}_{ij} \ge 1. $


\subsubsection{Analysis of \Cref{alg:weak-log}}
\label{sec:analysis-weak}

Now we prove the core technical lemma for this section:
\begin{lemma}
    \label{lem:residue_feasible}
    Consider an iteration with target load $L$: after the freezing
    steps in line~\ref{line:edge_setting},
    \begin{enumerate}[nosep,label=(\alph*)]
        \item \label{item:job-cond} For any job $j$ frozen in this round, we have
        $
            \sum_{i \in N(j)}\hat{x}_{ij} \geq 1.
        $
      \item \label{item:machine-cond} For any unfrozen machine
        $i \in A_M$, we have
        $ w_i + \sum_{j \in N(i)\cap A_J}x^*_{ij}s_{ij} \leq
        (1+7\eps)L$.
        Moreover, if not all jobs $j \in N(i)$ are frozen at the end
        of the iteration, then we have:
        $w_i + \sum_{j \in N(i)\cap A_J}x^*_{ij}s_{ij} \leq (1 + \varepsilon)L.$
    \end{enumerate}
\end{lemma}

\begin{proof}
  Let $A_J^-$ and $w_i^-$ be the active job set and the load of machine
  $i$ immediately before line \ref{line:edge_setting}, and $A_J$ and
  $w_i$ be the corresponding values after line
  \ref{line:set_setting}. Let $F = A_J^- \setminus A_J$ be the set of
  jobs frozen in this iteration.
By \Cref{lem:distributed-bfs},  for every job $j\in A_J^-$,
\[
  d_j=r
  \quad\Longleftrightarrow\quad
  \operatorname{dist}_{G[A_M\cup A_J^-]}(J_0,j)=2r
  \text{ and } r\le 1/\eps,
\]
and $d_j=+\infty$ if no such path of length at most $2/\eps$ exists.

  We first prove the job property~\ref{item:job-cond}. For a job $j$,
  we do a case analysis on $d_j$:
  \begin{enumerate}
  \item If $d_j = 0$, then we set $\hat{x}_{\cdot, j} = x'_{\cdot,
    j}$. Since $j$ was not frozen at previous search iteration, it
    returns \textsc{accept} in the previous round's call of
    \distYoung. Property~\ref{item:job-cond} then follows from the job
    side condition of \Cref{thm:dist-young}.
  \item If $d_j \geq 1$, then we have that $\hat{x}_{\cdot, j} \gets
    \eps d_j x^*_{\cdot, j} + (1-\eps d_j)x'_{\cdot, j}.$
    Similarly, since $j$ is not frozen last iteration, we have that
    $\sum_{i \in N(j)} x'_{ij} \geq 1.$ Also,
    $d_j \neq 0$ implies that $j \not \in J_0$, so job node $j$ returns \textsc{accept} in the
    current iteration's call of \distYoung. Therefore, we
    have $\sum_{i \in N(j)} x^*_{ij} \geq 1.$
    Thus,
    \begin{align*} \sum_{i \in N(j)}\hat{x}_{ij} &= \eps
      d_j \sum_{i \in N(j)} x^*_{ij} + (1-\eps d_j)\sum_{i \in
        N(j)} x'_{ij} \\ &\geq \eps d_j + (1-\eps d_j) \quad \geq \quad 1,
    \end{align*}
    which proves property~\ref{item:job-cond} in this case as well.
  \end{enumerate}

  We now turn to the machine property~\ref{item:machine-cond}, which
  we prove by induction over the algorithm execution. For the first
  search iteration with target load $L = n s_{\max}$, we know there
  are no frozen jobs since $\lambda^*\leq ns_{\max}$. Hence $w_i = 0$
  for all machines $i$. By the machine side condition of
  \Cref{thm:dist-young}, for every machine $i$, we have
  $\sum_{j \in N(i)} x^*_{ij}s_{ij} \leq (1 + \varepsilon) L$. Thus,
  property~\ref{item:machine-cond} holds for the first iteration.

Now we prove the inductive step. Consider any iteration after the first one.
Let the target load in the previous search iteration be
$L^+=(1+\eps)L$. Fix any machine $i\in A_M$ in the current
iteration. Since every machine with no active neighboring job is frozen at
the beginning of the iteration, $i$ has at least one active neighboring job
in $A_J^-$. Hence, by the strong part of the inductive hypothesis applied to
the previous iteration,
\[
    w_i^- + \sum_{j \in N(i)\cap A_J^-} x'_{ij}s_{ij}
    \leq (1 + \eps)L^+ \leq (1+3\eps)L .
\]
We have thus proved the following claim, where property \ref{item:x-star} follows from the machine-side condition of \Cref{thm:dist-young}.
  \begin{claim}
    \label{clm:about-x}
    The solutions $x'$ and $x^*$ satisfy for every machine $i \in A_M$:
    \begin{enumerate}[nosep,label=(\roman*)]
    \item \label{item:x-slash}
      $w^-_i + \sum_{j \in N(i) \cap A^-_J}x'_{ij}s_{ij} \leq (1 + 3\eps)L.$
    \item 
      \label{item:x-star}
      $w^-_i + \sum_{j \in N(i)}x^*_{ij}s_{ij} \leq (1 + \eps)L.$    
    \end{enumerate}
  \end{claim}

  Fix a machine $i$, and let $R \coloneqq N(i)\cap A_J^-$. We first
  consider the case when $R\cap F=\emptyset$. In this case, since no jobs from $R$ frozen in this iteration, we have $w_i = w^-_i$. Property~\ref{item:machine-cond} then directly
  follows from \Cref{clm:about-x}\ref{item:x-star} : \[ w_i + \sum_{j
    \in N(i)\cap A_J}x^*_{ij}s_{ij} = w^-_i + \sum_{j \in N(i)\cap
    A_J}x^*_{ij}s_{ij} \leq (1  +\eps) L.
  \]
  Therefore it suffices to consider the case when $R\cap F \neq
  \emptyset.$ Let $d_{\min}:=\min_{j\in R\cap F} d_j$ be the smallest
  label among all jobs in $R\cap F$ and let $j_{\min}\in R\cap F$ be a
  job with label $d_{\min}$.

  If $d_{\min}<1/\eps$, then every job in $R$ must be frozen in
  the current iteration. Indeed, for any job $j\in R$, since both
  $j_{\min}$ and $j$ are adjacent to the same machine $i$, a shortest
  path from $J_0$ to $j_{\min}$ followed by the two edges $(j_{\min},i)$
  and $(i,j)$ gives a path from $J_0$ to $j$ of length at most
  \[
  2d_{\min}+2 = 2(d_{\min}+1) \leq 2/\eps.
  \]
  Thus $j$ receives a finite label at most $d_{\min}+1$, and therefore
  $j\in F$. Moreover, every job $j\in R$ has label in
  $\{d_{\min},d_{\min}+1\}$. On the other hand, if
  $d_{\min}=1/\eps$, then every frozen job in $R\cap F$ has label
  exactly $1/\eps$.  Every unfrozen job in $R\cap A_J$ has
  distance exactly $2/\eps + 2$ from a job in $J_0.$
  We now do a case analysis on $d_{\min}$.
  \begin{enumerate}
  \item Case I: $d_{\min} < 1/\eps$. As shown above, this implies that all jobs $j \in N(i)\cap
    A_J^-$ belong to $F$ and all these jobs have labels in
    $\{d_{\min}, d_{\min} + 1\}$. Take an edge $(i, j)$:
    \begin{itemize}
    \item If $d_j = d_{\min}$, then $\hat{x}_{ij} = (\eps\,
      d_{\min} ) x^*_{ij} + (1 - \eps\, d_{\min}) x'_{ij}$.
    \item If $d_j = d_{\min} + 1$, then $\hat{x}_{ij} = \eps\,(d_{\min} + 1) x^*_{ij} + (1 - \eps\,(d_{\min} + 1)) x'_{ij}$.
    \end{itemize}
    In both cases, $\hat{x}_{ij} \leq \eps\,(d_{\min} + 1) \, x^*_{ij}
    + (1 - \eps\, d_{\min}) x'_{ij}.$
    As a result, 
    \begin{align*}
w_i
=
w_i^-+\sum_{j\in R}\hat{x}_{ij}s_{ij}
\le
w_i^-+\sum_{j\in R}\eps(d_{\min}+1)x^*_{ij}s_{ij}
+\sum_{j\in R}(1-\eps d_{\min})x'_{ij}s_{ij}.
    \end{align*}
    Using both parts of \Cref{clm:about-x} to simplify the RHS above,
    we get:
    \[
    \begin{aligned}
      w_i - w_i^-\quad 
      &\leq \quad
      \eps(d_{\min}+1)((1 + \eps)L-w_i^-)
      +
      (1-\eps d_{\min})
      ((1+3\eps)L-w_i^-) \\
      &\leq \quad (1 + \eps)((1+3\eps)L-w_i^-)
      \leq \quad (1 + 7\eps)L - w^-_i.
    \end{aligned}
    \]
    Therefore, 
    $ 
    w_i \leq (1 + 7\eps)L.$ 
    To prove property~\ref{item:machine-cond}, observe that since $R \cap A_J = \emptyset$, 
    \begin{align*}
      w_i + \sum_{j \in R\cap A_J}x^*_{ij}s_{ij} = w_i \leq (1 + 7\eps)L.
    \end{align*}
  \item Case II: $d_{\min} = 1/\eps.$ In this case, as shown above, \emph{every} $j \in F \cap R$ has $d_j = d_{\min} = 1/\eps$, so by Line \ref{line:edge_setting} we have  $\hat{x}_{ij} = x^*_{ij}$ for all $j \in F\cap R.$
    Hence,
    \begin{align*}
      w_i + \sum_{j \in N(i)\cap A_J}x^*_{ij}s_{ij} &= w^-_i + \sum_{j \in N(i)\cap F}x^*_{ij}s_{ij}+\sum_{j \in N(i)\cap A_J}x^*_{ij}s_{ij} \\
      &= w^-_i + \sum_{j \in N(i)} x^*_{ij}s_{ij} 
      \quad \leq \quad (1 + \varepsilon)L,
    \end{align*}
    where the last inequality uses
    \Cref{clm:about-x}\ref{item:x-star}. This proves
    property~\ref{item:machine-cond} for this case.
  \end{enumerate}
  Having proved both properties~\ref{item:job-cond}
  and~\ref{item:machine-cond}, we have \Cref{lem:residue_feasible}. 
\end{proof}

\begin{corollary}
  \label{cor:almost_satisfy}
  For any job $j$, we have $\sum_{i \in N(j)}\hat{x}_{ij} \geq
  1$ at the end of the algorithm.
\end{corollary}
\begin{proof}
  By \Cref{lem:residue_feasible}, any job $j$ frozen in the search
  iteration has $\sum_{i \in N(j)} \hat{x}_{ij} \geq 1$.  If a job $j$ is
  not frozen in the search iteration, then we set $\hat{x}_{\cdot, j}
  \gets x^*_{\cdot, j}.$ The statement then follows from the job side
  condition of \Cref{thm:dist-young}.
\end{proof}

\begin{corollary}
  \label{cor:almost_same_stop_point}
  Let $L^*$ be the largest value of $L$ for which some job node returns
  \textsc{reject} during the call to \distYoung in our algorithm. Then, $L^* < \lambda^*$, where $\lambda^*$ is the maximum load of the optimal integral solution. 
\end{corollary}
\begin{proof}
Note that our algorithm only freezes nodes due to some job returning \textsc{reject}. Also, $w_i$ can only become non-zero due to some job becoming frozen. So by definition of $L^*$, when the algorithm begins the execution of iteration $L^*$, no node has been frozen in previous iterations, and we have: $A_M = M$, $A_J = J$, and $w_i = 0$ for all machines $i$. Thus, the MixPC LP in Line \ref{line:mpc} is precisely the feasibility LP for load balancing with maximum load $L^*$. By \Cref{lem:mixpc-black-box}, the fact that some job node returns \textsc{reject} implies that this LP is infeasible, so $L^* < \lambda^*$.
\end{proof}

\begin{corollary}
  \label{cor:low_load}
  For each machine $i$ that is removed from $A_M$ at
  Line~\ref{line:freeze-machine}, let $L_i$ be the target load $L$ in
  that iteration. Then, after the search iterations, the load $w_i$ of
  machine $i$ due to frozen jobs satisfies
  \[
    w_i \leq (1+O(\eps))L_i \leq (1+O(\eps))L^*.
  \]
\end{corollary}
\begin{proof}
  Fix a machine $i$ that is removed from $A_M$ at
  Line~\ref{line:freeze-machine}, and let $L_i$ be the target load in
  that iteration. If this is the first search iteration, then $w_i=0$
  and the claim is trivial. Otherwise, let
  $L_i^+=(1+\eps)L_i$ be the target load in the immediately previous
  iteration.

  Let $A_J^-$ and $w_i^-$ denote the active job set and the load of $i$
  just before the freezing step of the current iteration. Since $x'$ is
  the solution from the previous iteration, the previous iteration's
  invariant \Cref{clm:about-x}\ref{item:x-star} gives
  \[
    w_i^-+\sum_{j\in N(i)\cap A_J^-}s_{ij}x'_{ij}
    \le (1+O(\eps))L_i^+.
  \]
  Although $i$ may receive load in the current freezing step, this
  load can only come from the $x'$-part of the interpolation: since
  $i\notin A_M$, we have $x^*_{ij}=0$ for all $j\in A_J^-$, and hence
  \(
    \hat{x}_{ij}
    =
    \eps d_jx^*_{ij}+(1-\eps d_j)x'_{ij}
    \le x'_{ij}.
  \)
  Therefore, if $F\subseteq A_J^-$ is the set of jobs frozen in this
  iteration, then
  \[
    w_i
    \le
    w_i^-+\sum_{j\in N(i)\cap F}s_{ij}\hat{x}_{ij}
    \le
    w_i^-+\sum_{j\in N(i)\cap A_J^-}s_{ij}x'_{ij}
    \le
    (1+O(\eps))L_i^+
    =
    (1+O(\eps))L_i.
  \]
  After this iteration, all future $x^*$ and $x'$ values on edges
  incident to $i$ are zero (because $i$ has been removed from $A_M$), so $i$ receives no further load.
  Finally, by the definition of $L^*$, all freezing steps occur at target
  loads at most $L^*$. Hence,
  $
  w_i
  \leq
  (1+O(\eps))L_i
  \leq
  (1+O(\eps))L^*.
  $
\end{proof}

Now we are ready to establish \Cref{thm:weak-log}.
\begin{proof}[Proof of~\Cref{thm:weak-log}]
After Line~\ref{line:unfrozen_setting}, fix any machine $i$. If $i$ has
been removed from $A_M$, then
\Cref{cor:almost_same_stop_point,cor:low_load} bounds its load by
$(1+O(\eps))\lambda^*$, and the final assignment of the remaining active
jobs adds no load to $i$. Otherwise, $i$ remains in $A_M$, and applying
\Cref{lem:residue_feasible} to the final search iteration gives
\[
  w_i+\sum_{j\in N(i)\cap A_J}x^*_{ij}s_{ij}
  \le (1+O(\eps))L_{\mathrm{last}}
  \le (1+O(\eps))\lambda^*,
\]
where \(L_{\mathrm{last}}\le (1+\eps)s_{\min}\le (1+\eps)\lambda^*\).
Thus every machine has load at most $(1+O(\eps))\lambda^*$. By
\Cref{cor:almost_satisfy}, every job receives at least one unit of mass,
so the output is feasible.

  We now show that $s_{ij} \leq (1 + \eps)\lambda^*$ if the edge
  $(i,j)$ is in the support of the solution $\hat{x}$ (as required by Definition \ref{def:frac-LB}). Fix any such
  edge $(i, j)$ with $\hat{x}_{ij} > 0$. If $j$ is frozen during the search iterations at target load $L$, then
any support edge coming from the $x^*$ part lies in $E_L$, while any
support edge coming from the $x'$ part lies in $E_{(1+\eps)L}$. Hence
every support edge used by $j$ has size at most $(1+\eps)L$. Since
freezing occurs only at target loads $L\le L^*$, by
\Cref{cor:almost_same_stop_point},
\[
  s_{ij}\le (1+\eps)L\le (1+\eps)L^*\le (1+\eps)\lambda^*.
\] If $j$ is not frozen, then it is assigned using the final call, whose
threshold is at most $(1+\eps)s_{\min}$. Since
$s_{\min}\le \lambda^*$, this also gives
$s_{ij}\le (1+\eps)\lambda^*$.

At last, we prove that \Cref{alg:weak-log} works in
$\poly(\eps^{-1}\,\log(mn \cdot
\frac{s_{\max}}{s_{\min}}))$ rounds. All initialization steps  take a constant number of rounds. At each search
iteration, \distYoung takes $\poly(\log(mn)/\eps)$ rounds,
the call to \distBFS takes $O(1/\eps)$ rounds by \Cref{lem:distributed-bfs} $O(1/\eps)$ rounds, and all other
steps take $O(1)$ rounds. Multiplying by $T =
\lceil\log_{1+\eps} (n \cdot s_{\max}/s_{\min})\rceil =
\log(n \cdot s_{\max}/s_{\min})/\eps$, the number of search
iterations, we get the claimed bound of $\poly(\eps^{-1} \, \log(mn\cdot s_{\max}/s_{\min}))$ rounds.
\end{proof}


\section{A Strongly Logarithmic Round Complexity}
\label{sec:strongly-log}

We now introduce our procedure using a remove the dependence on $\log (s_{\max}/s_{\min})$
in \Cref{thm:weak-log}, and hence complete our proof of~\Cref{thm:intro-lb}.



\subsection{Algorithm Overview}

The main difficulty in obtaining a strongly polynomial algorithm is
that \Cref{thm:weak-log} assumes a bounded and known range of edge
sizes, whereas in the original instance the ratio between the largest
and smallest edge size may be arbitrarily large and is not given
globally.  The algorithm \texttt{StrongLoadBalancer} removes this
assumption by applying a shifting scheme to logarithmic size levels.
It constructs several bounded-range subinstances, solves each of them
using \Cref{thm:weak-log}, and then uses \Cref{thm:dist-young} once
more to merge the candidate supports into the final fractional
assignment.

We first classify every edge $(i, j) $ as having level
$\ell_{ij}=\left\lceil \log_{1+\eps} s_{ij}\right\rceil.$ For each
job $j$, let $L_{\min}(j)=\min_{i\in N(j)} \ell_{ij}$ be the level of
the smallest size it can take, and define its interval
$I_j=[L_{\min}(j),L_{\min}(j)+\beta-1],$ where
$\beta = \lceil \log_{1+\eps}(mn/\eps)\rceil$.  This interval $I_j$
contains all edge levels that are ``relevant'' for job $j$: edges
whose level is larger than the upper endpoint of $I_j$ are locally
deleted.  Intuitively, such an edge is dominated by a much lighter
edge incident to the same job, so any solution using this heavy edge
can be modified to use the lighter edge while increasing the maximum
load only by $(1 + \varepsilon)$.

Next, the algorithm applies a \emph{shifted bucket decomposition} to
the logarithmic size levels.  Define the bucket width to be
\( B := \nicefrac{2\beta}{\eps}.  \) For every offset
$a\in\{0,\ldots,B-1\}$, we consider buckets with shifted boundaries.
The $b^{th}$ \emph{full bucket} corresponding to offset $a$ would be
\(
  [\,bB+a,\,(b+1)B+a\,),
\)
but the algorithm only uses its \emph{active} part
\[
  \BB_{b,a}
  :=
  [\,bB+a,\,(b+1-\eps)B+a) = [\,bB+a, \,(b+1)B+a - C\,),
\]
where we shrink it from the right end by $C :=\eps B=2\beta$; this
ensures that consecutive active buckets are separated by a ``guard
band'' of length $C$.

A job $j$ is called \emph{active} in the instance for offset $a$ if
its interval $I_j$ is contained in some active bucket $\BB_{b,a}$ for
some $b \in \ZZ$; else, $j$ is called inactive for this offset. As the
offset $a$ ranges over $[0, B -1]$, the job $j$ fails to be active
only when $I_j$ intersects one of the guard bands, which happens for
only an $O(\beta/B) = O(\eps)$ fraction of the offsets. We use this
fact in proving that the final mixed packing-covering instance is
feasible.

Fix one offset $a$; for each active job $j$, let $b_j$ be the index
such that $I_j \sse \BB_{b_j, a}$. Each machine $i$ now considers the
active jobs adjacent to it, and joins the largest bucket it sees:
\[
  b_{\max}(i) := \max\{b_j \mid j\in N(i),\ j\text{ active}\}.
\]
Finally, if machine $i$ is adjacent to an active job $j$ with
$b_j<b_{\max}(i)$, then the algorithm fully assigns job $j$ to one 
such machine $i$ in the solution for this offset, and removes $j$ from
further consideration. We show that the gap of length $C$ between two
active buckets means that this assignment increases the maximum load
by at most a $(1 + \varepsilon)$ factor.

After this cleanup, all remaining edges $(i,j)$ satisfy
$b_j = b_{\max}(i)$.  Hence, for this fixed offset $a$, the remaining
graph decomposes into subinstances that are both vertex- and
edge-disjoint. Moreover, for every bucket $b$, the lower bound and the
upper bound of the edge sizes are known to every node participating in
it. Thus, the algorithm can call \textsc{WeakFractLoadBalancer} on
this bucket with
\[
  s_{\min}(b,a)=(1+\eps)^{bB+a-1},
  \qquad
  s_{\max}(b,a)=(1+\eps)^{(b+1-\eps)B+a-1}.
\]
The bucket subinstances are disjoint, so all those calls for the same
offset can be executed in parallel.  Repeating this for all offsets
produces a collection of offset solutions $x^a$.

Finally, the algorithm merges all offset solutions through one mixed
packing-covering instance.  Let $E'$ be the union of all edges that
receive positive value in at least one offset solution.  For each
machine, we choose a local capacity from the maximum load it receives
in any offset solution and from the largest edge of $E'$ incident to
it.  The analysis shows that, because each job is active in almost all
offsets, these capacities make the mixed packing-covering instance
feasible.  The algorithm then runs \distYoung on this final instance
and outputs the resulting fractional assignment.

\subsection{Algorithm Description}
\label{sec:strong-algo-desc}

The main algorithm (given in \Cref{alg:strong-log}) follows the
outline above: it calls \CleanEdges{} to compute edge levels,
construct the relevant interval $I_j$ for every job $j$, and
delete incident edges whose levels do not lie in $I_j$. Then, for
each offset $a \in [B]$ (synchronously)
it calls \BuildGraph{} which constructs the disjoint subinstances for
each bucket $\BB_{b,a}$, while pre-assigning some of the jobs.
These instances are solved independently and in
parallel. Finally, it calls \SetCapacitiesAndSolve{}, which sets the
machine capacities, constructs the corresponding \MixPC instance, and
solves it using \distYoung to obtain the solution $\hat{x}$ described
in the overview. Finally, the algorithm does some clean-up steps---it
normalizes the job assignments to ensure that each job is fully
assigned, potentially increasing machine loads by at most a factor
of $1+O(\varepsilon)$.
  
\begin{algorithm}[H]
  \caption{\texttt{StrongLoadBalancer}}
  \label{alg:strong-log}
  \KwIn{Each node knows $\varepsilon$. Each machine node $i$ knows
    $\forall j \in N(i)s_{ij}$. \\ $\qquad\qquad$ Each job node $j$
    knows $\forall i \in N(j), s_{ij}.$} \KwOut{Each machine $i$
    outputs values $\hat{x}_{ij}$ for all $j \in N(i)$ satisfying
    \Cref{thm:strong-log}.}

  $\beta \gets \lceil \log_{1+\eps}(nm/\eps)\rceil$.

  \CleanEdges{} \tcp*{Computes interval $I_j$ for each job $j$}

  $B \gets 2\beta/\eps$\;

  \For{$a \gets 0$ \KwTo $B-1$}{
    \lForEach(\tcp*[f]{$x^a$ is the solution for the instance piece $a$}){$i \in M, j \in N(i)$}{$x^a_{ij}\gets 0$}
    \BuildGraph{$a$} 
    \tcp*{Formulate the graph for offset $a$; disjoint graphs for different buckets.}
    
    \ForEach(\tcp*[f]{ solve disjoint sub-instances in parallel}){bucket $b \neq -\infty$ \textnormal{\emph{in parallel}}}{
      $V_b \gets \{\, j : b_j = b \,\}\ \cup\ \{\, i : b_{\max}(i) = b \,\}$ \tcp*{$V_b$ locally represented\ via \Cref{def:local-sets}}
      $E_b \gets \{\, (i,j) \in E : i \in V_b,\ j \in V_b \,\}$\tcp*{$E_b$ locally represented\ via \Cref{def:local-sets}}
      $s_{\min}(b,a) \gets (1+\eps)^{\,bB+a-1}$;\quad
      $s_{\max}(b,a) \gets (1+\eps)^{\,(b+1-\eps)B+a-1}$
      
      $x^a\!\restriction_{E_b} \gets$
      \texttt{WeakFractLoadBalancer}$((V_b, E_b),\, s_{\min}(b,a),\,
      s_{\max}(b,a), \varepsilon)$ 
    }
  }
  
  \SetCapacitiesAndSolve{}
  
  \lForEach{job $j$} { $v_j\gets \sum_{i\in N(j)}\hat{x}_{ij}$}
  \lForEach(\tcp*[f]{Normalize job $j$'s assignment to total mass $1$}){machine $i \in
    N(j)$}{$\hat{x}_{ij}\gets \hat{x}_{ij} /
    v_j.$ \label{line:round-up-job} }
  \lForEach{machine $i$}{output
  $\hat{x}_{ij}$ for all $j \in N(i)$}
\end{algorithm}

\medskip\textbf{Procedure \CleanEdges{}:} it computes the edge levels
$\ell_{ij}$, the intervals $I_j$ for each job, and removes some edges;
\Cref{lem:strong-cleaning} shows that this edge-deletion increases the
optimal load by at most
$(1+O(\eps))$. 

\begin{procedure}[H]
  \caption{CleanHeavyEdges()}
  \label{proc:clean-edges}
  \ForEach{job $j$}{
    \lForEach{$i \in N(j)$}{
      $\ell_{ij} \gets \lceil \log_{1+\eps} s_{ij}
      \rceil$ }
    $L_{\min}(j) \gets \min_{i \in N(j)} \ell_{ij}$ and
    $L_{\max}(j) \gets L_{\min}(j) + \beta - 1$ \; 
    $I_j \gets [\,L_{\min}(j),\, L_{\max}(j)\,]$\;
    \ForEach{$i \in N(j)$}{
      \lIf{$\ell_{ij} > L_{\max}(j)$}{$E \gets E \setminus \{(i,j)\}$
        \tcp*[f]{$i$ and $j$ locally mark this edge as
          unusable} \label{line:delete_edge}} 
    }
  }
\end{procedure}

\medskip\textbf{Procedure \BuildGraph{}}: For the offset $a$, each job
node $j$ checks if $I_j \sse \BB_{b_j, a}$ for some $b_j \in \ZZ$; if
not, the job is inactive for this iteration (indicated by
$b_j = -\infty$). Each machine $i$ then chooses $b_{\max}(i)$, the
maximum bucket label among all adjacent active jobs.  Finally, the procedure
pre-assigns some jobs: if the set
\( C_j := \{\, i \mid j \in N(i),\ b_j < b_{\max}(i) \,\} \) is
non-empty, the algorithm assigns $j$ to an some machine $i_j \in C_j$.

\begin{procedure}[H]
  \caption{BuildGraph($a$)}
  \label{proc:build-graph}
  \ForEach{job $j$}{
    \leIf{$I_j \subseteq \BB_{b,a}$ \textnormal{for some} $b$}
    {$b_j \gets b$}
    {$b_j \gets -\infty$}
  }
  \ForEach{machine $i$}{
    $b_{\max}(i) \gets \max_{j \in N(i)} b_j$
  }
  \ForEach{job $j$}{
    $C_j \gets \{\, i : j \in N(i),\ b_j < b_{\max}(i),\ b_j \neq -\infty \,\}$\;
    \lIf{$C_j \neq \emptyset$}{
      choose arbitrary machine $i_j \in C_j$, and set 
      $x^a_{i_j,j} \gets 1$, $b_j \gets -\infty$ \label{line:assign_jobs}
    }
  }
\end{procedure}

\medskip\textbf{Procedure \SetCapacitiesAndSolve{}:} It constructs the union $E'$ of the supports for the solutions $x^a$ for various offsets $a$. Each machine node $i$ computes 
the largest edge size $W_i$ incident to it in $E'$, and the
maximum load $\Lambda_i$ placed on machine $i$ by any of these solutions. It sets each machine's target load
$L'_i \gets (1 + O(\varepsilon))\,\max(W_i, \Lambda_i)$, and runs
\distYoung on this instance with accuracy parameter $\eps$ to get the
final solution $\hat{x}.$ 

\begin{procedure}[H]
\caption{SetCapacitiesAndSolve()}
\label{proc:set-capacity-and-solve}
$E' \gets \{(i,j)\in E : x^a_{ij}>0 \text{ for some } a\}$\;
        \ForEach{machine $i$}{
        $W_i \gets \max\{s_{ij} : (i,j)\in E'\}$
            \tcp*{heaviest edge ever used on $i$: the max edge weight}     
        $\Lambda_i \gets \max_{0 \le a \le B-1}\ \sum_{j \in N(i)} s_{ij}\, x^a_{ij}$
            \tcp*{heaviest per-offset load on $i$}
        $L'_i \gets (1+O(\eps))\,\max\big(\Lambda_i,\ W_i\big)$
            \tcp*{packing capacity}
    }
    Let $\mathcal{M}$ be the following \MixPC instance on $G' = (V, E')$:
    \begin{align*}
      \textstyle \sum_{i \in N(j)} \hat{x}_{ij} &\ge 1 & \qquad \qquad &\forall   j \\
      \textstyle \sum_{j \in N(i)} s_{ij}\,\hat{x}_{ij} &\le L'_i&&\forall i
    \end{align*}
    
    $\hat{x} \gets \distYoung(V, E', \eps)$
\end{procedure}

\subsection{The Analysis}


\begin{lemma}[Cleaning Loss]
  \label{lem:strong-cleaning}
  Given an instance $\calI$, let $\calI'$ be the instance after
  \CleanEdges{}, with optimum integral
  makespan $\lambda^*(\calI')$. Then $\lambda^*(\calI') \le (1+O(\eps))\,\lambda^*(\calI)$.
\end{lemma}

\begin{proof}
  Consider an optimum integral assignment $\sigma$; If $(\sigma(j),j)$
  survives \CleanEdges{}, keep $j$ unchanged.  Otherwise assign $j$ to
  a minimum-level neighbor $i_{\min}(j)$. Since the deleted edge is at
  least $\beta$ levels above the minimum level of $j$,
  \[
    s_{i_{\min}(j),j}
    \le
    (1+\eps)^{-(\beta-1)}\cdot s_{\sigma(j),j}
    \le
    O\!\left(\frac{\eps}{mn}\right)\lambda^*(\calI) .
  \]
  Every machine receives at most $n$ reassigned jobs, giving an extra
  load of $O(\eps/m) \cdot \lambda^*(\calI)$, which proves the claim.
\end{proof}

We also define an auxiliary instance $\bar I_a$: this is the instance we
would obtain if line~\ref{line:assign_jobs} were skipped.  Equivalently,
$\bar I_a$ is the subinstance of $I'$ induced by the jobs that are active for
offset $a$, with all surviving edges from \CleanEdges{} kept.  Let
$\bar\lambda$ be the optimum integral makespan of $\bar I_a$.

\begin{lemma}[Bucketing Loss]
  \label{lem:pre-offset}
  Let $\calI_a$ be the instance produced by $\BuildGraph(a)$, where the
  jobs assigned in line~\ref{line:assign_jobs} are fixed to their
  chosen machines.  Then,
  $\lambda^*(\calI_a) \leq (1+O(\eps))\lambda^*(\calI')$.
\end{lemma}

\begin{proof}
  First, we observe that dropping inactive jobs only causes the
  optimal load to decrease, so it remains to bound the increase in
  load due to fixing the job assignment in
  line~\ref{line:assign_jobs}.
  Suppose such a job $j$ (in bucket $\BB_{b_j,a}$) is assigned to
  machine $i_j$. Since $i_j\in C_j$, there must be an active job
  $j'\in N(i_j)$ with
  \[
    b_{j'} = b_{\max}(i_j) \ge b_j + 1 .
  \]
  However, the interval $I_{j'}$ for job $j'$ is contained within the
  bucket $\BB_{b_{j'},a}$, so every edge incident to this job $j'$ has
  level at least $b_{j'}B+a$. This means, 
  \[
    \lambda^*(\calI') \ge \min_{i\in N(j')}s_{ij'} \ge
    (1+\eps)^{b_{j'}B+a-1} \ge (1+\eps)^{(b_j+1)B+a-1}.
  \]
  On the other hand, since job $j$ lies in bucket $b_j$, its size on
  machine $i_j$ is at most $(1+\eps)^{(b+1-\eps)B+a-1}$. Putting these
  facts together,
  \begin{gather}
    s_{i_j j}
    \le
    (1+\eps)^{-\eps B} \lambda^*(\calI')
    =
    (1+\eps)^{-2\beta} \lambda^*(\calI') \leq (\eps/mn)^2
    \lambda^*(\calI'). \label{eq:guard-band}
  \end{gather}
  Now each machine receives at most $n$ directly assigned jobs, so the
  extra load from all direct assignments on any machine can be (very loosely
  bounded) by $O(\eps)\,\lambda^*(\calI')$, which completes the proof.
\end{proof}



\begin{lemma}[Support Lemma]
  \label{lem:offset-support-bound}
  For every offset $a$ and every edge $(i,j)$ with $x^a_{ij}>0$, we
  have \( s_{ij}\le (1+O(\eps))\lambda^*(\calI).  \)
\end{lemma}

\begin{proof}
  Fix an offset $a$ and an edge $(i,j)$ with $x^a_{ij}>0$. First,
  suppose that $(i,j)$ is created by the direct assignment step in
  line~\ref{line:assign_jobs}. Then the guard-band argument
  in~(\ref{eq:guard-band}) from the proof of \Cref{lem:pre-offset}
  gives $s_{ij} \le (\eps/mn)^2\lambda^*(\calI')$.  By
  \Cref{lem:strong-cleaning}, this is at most
  $(1+O(\eps))\lambda^*(\calI)$.

  Otherwise, $(i,j)$ belongs to the support returned by
  \texttt{WeakFractLoadBalancer} on some bucket subinstance $H_{a,b}$.
  By the support guarantee of \Cref{def:frac-LB} and
  \Cref{thm:weak-log}, we have that 
  \(
    s_{ij}\le (1+O(\eps))\lambda^*(H_{a,b}).
  \)
  Since $H_{a,b}$ is one of the remaining bucket subinstances of
  $\calI_a$, restricting an optimal assignment for $\calI_a$ to this
  bucket gives a feasible assignment for $H_{a,b}$. Therefore
  \(
  \lambda^*(H_{a,b})\le \lambda^*(\calI_a).
  \)
  Using \Cref{lem:strong-cleaning,lem:pre-offset}, we get
  \[
    s_{ij}
    \le (1+O(\eps))\lambda^*(H_{a,b})
    \le (1+O(\eps))\lambda^*(\calI_a)
    \le (1+O(\eps))\lambda^*(\calI')
    \le (1+O(\eps))\lambda^*(\calI).
  \]
  This proves the lemma.
\end{proof}

\begin{lemma}
  \label{lem:final-mpc-feasible}
  In \SetCapacitiesAndSolve{}, the \MixPC instance $\mathcal{M}$ is
  feasible, and every machine $i$ has target load
  $L'_i \leq (1+O(\eps))\lambda^*(\calI)$.
\end{lemma}

\begin{proof}
  Using \Cref{lem:strong-cleaning,lem:pre-offset} and the fact that we
  find a $(1 + O(\eps))$-approximate solution $x^a$ for each offset
  instance $\calI_a$, the load assigned by $x^a$ on each machine is at
  most $(1+O(\eps))\lambda^*(\calI)$. Since $\Lambda_i$ is the largest
  of these loads, we have $\Lambda_i \leq (1+O(\eps))\lambda^*(\calI)$
  for every machine $i$. Also, by \Cref{lem:offset-support-bound},
  every edge in $E'$ has size at most
  $(1+O(\eps))\lambda^*(\calI)$. Therefore, for every machine $i$,
  $W_i=\max\{s_{ij} \mid (i,j)\in E'\} \le
  (1+O(\eps))\lambda^*(\calI).$ Therefore, the target load
  $L'_i=(1+O(\eps))\max(\Lambda_i,W_i)\leq
  (1+O(\eps))\lambda^*(\calI)$.

  It remains to show that the instance is feasible. Define the
  ``average solution''
  $\bar{x}_{ij}:= \frac{1}{B}\sum_{a=0}^{B-1}x^a_{ij}$. By
  construction, $\bar{x}$ is supported on $E'$. Moreover, since there
  are $B$ offsets, and each job is inactive in at most
  $C+ \beta \leq 3\beta$ of them, it participates in at least a
  $(1-3\eps)$ fraction of the offsets. Thus
  $\sum_i \bar{x}_{ij}\geq 1-3\eps$ for every job $j$. Now rescaling
  the solution up can cause the load to increase by at most a factor
  of $(1+O(\eps))$ more, proving the claim.
\end{proof}

\begin{restatable}[Fractional Load Balancing;
 Strongly Polynomial]{theorem}{stronglog}
 \label{thm:strong-log}
 There is a deterministic distributed $(1+O(\eps))$-approximation
 algorithm for the Fractional Load Balancing problem
 (Defn.~\ref{def:frac-LB}) without any additional global information
 that works in $\poly(\log(mn)/\eps)$ rounds in the \congest model,
 and treats \Cref{lem:mixpc-black-box} and \Cref{thm:weak-log} as black
 boxes.
\end{restatable}

\begin{proof}[Proof of \Cref{thm:strong-log}]
  Run \Cref{alg:strong-log} with internal accuracy parameter
  $\eps_0=\eps/c$, for a sufficiently large constant $c$. We prove a
  $(1+O(\eps_0))$ guarantee; choosing $c$ large enough gives the
  stated $(1+\eps)$ guarantee.

  By \Cref{lem:final-mpc-feasible}, the \MixPC instance $\mathcal{M}$
  constructed in \SetCapacitiesAndSolve{} is feasible, and every
  machine satisfies $L'_i\le (1+O(\eps_0))\lambda^*$. Since
  $\mathcal{M}$ is feasible, \distYoung returns a solution $z$ with
  $\sum_i z_{ij}\ge 1/(1+O(\eps_0))$ for every job $j$, and
  $\sum_j s_{ij}z_{ij}\le (1+O(\eps_0))L'_i$ for every machine
  $i$. The final normalization step scales each job by at most
  $1+O(\eps_0)$, so the output $\hat{x}$ satisfies
  $\sum_i \hat{x}_{ij}=1$ for every job $j$, and for every machine
  $i$,
  $\sum_j s_{ij}\hat{x}_{ij}\le (1+O(\eps_0))L'_i\le
  (1+O(\eps_0))\lambda^*$.

  Moreover, the normalization step creates no new support edge. Thus
  if $\hat{x}_{ij}>0$, then $(i,j)\in E'$. By
  \Cref{lem:offset-support-bound}, every edge in $E'$ has size at most
  $(1+O(\eps_0))\lambda^*$. Hence, by carefully choosing $c$, the
  algorithm outputs a feasible fractional assignment of load at most
  $(1+\eps)\lambda^*$, with no positive mass on edges heavier than
  $(1+\eps)\lambda^*$.

  It remains to bound the number of rounds. \CleanEdges{} takes $O(1)$
  rounds.  For a fixed offset $a$, \BuildGraph$(a)$ also takes $O(1)$
  rounds: jobs send their bucket labels, machines compute and return
  $b_{\max}(i)$, and jobs do the direct assignments. After
  \BuildGraph$(a)$, the remaining bucket instances are edge- and
  vertex-disjoint, so all calls to \texttt{WeakFractLoadBalancer} for
  this offset run in parallel. In each bucket, the ratio between the
  largest and smallest edge sizes is at most $(1+\eps_0)^{O(B)}$, so
  the weak algorithm takes $\poly(\log(mn)B/\eps_0)$ rounds. Since
  $\beta=O(\log(mn)/\eps_0)$ and $B=O(\log(mn)/\eps_0^{2})$, one
  offset takes $\poly(\log(mn)/\eps)$ rounds. There are
  $B=\poly(\log(mn)/\eps)$ offsets, so all offset computations
  together take $\poly(\log(mn)/\eps)$ rounds. Finally,
  \SetCapacitiesAndSolve{} only does local aggregation followed by one
  call to \distYoung, which also takes $\poly(\log(mn)/\eps)$
  rounds. Thus the total round complexity is $\poly(\log(mn)/\eps)$ in
  the \congest{} model.
\end{proof}


\section{Rounding to an Integral Solution}
\label{sec:rounding}

In this section, we observe that the fractional solution produced by \Cref{thm:strong-log} can be rounded to an integer solution with a loss of a factor of $(2+\eps)$. The ideas here are fairly standard; in particular, the desired rounding procedure follows from translating the near-linear-time algorithm of Li~\cite{Li23} to the \congest setting. To begin, let us formally define our goal.

\begin{defn}[Distributed $\alpha$-Approximation for Integer LB]\label{def:dilb-approx}
  A distributed protocol for Integer LB is an $\alpha$-approximation
  if for each distributed instance $\calI$ of LB as in
  \Cref{def:dilb}, it terminates in finitely many rounds, with each
  machine node $i$ returning a set $S_i \subseteq N(i)$ indicating the
  jobs assigned to machine $i$, such that
  \[ \sum_{j \in S_i} s_{ij} \leq \alpha\, \lambda^*(\calI)\] for all
  $i$, and each job $j$ belongs to exactly one output set $S_i$ where
  $i \in N(j)$.
\end{defn}

\begin{restatable}[Rounding Algorithm]{theorem}{rounding}
  \label{thm:rounding}
  There is a deterministic distributed $(2+O(\eps))$-approximation
  algorithm for the Integer Load Balancing problem
  (Defn.~\ref{def:dilb-approx}), which runs in $\poly(\log(mn)/\eps)$
  rounds in the \congest model, when given a
  $(1+O(\eps))$-approximation for the Fractional Load Balancing
  problem as input.
\end{restatable}

For completeness, a proof of this result appears in \Cref{app:rounding}. Combining the rounding result in \Cref{thm:rounding} with our result in \Cref{thm:strong-log} obtaining fractional load-balancing solutions, we get our main result for integer load balancing:

\begin{restatable}[Integral Load Balancing]{theorem}{lbfinal}
  \label{thm:lb-final}
  There is a deterministic distributed $(2+\eps)$-approximation
  algorithm for the Integer Load Balancing problem
  (Defn.~\ref{def:dilb-approx}), which runs in $\poly(\log(mn)/\eps)$
  rounds in the \congest model.
\end{restatable}


\section{Distributed Mixed Packing and Covering}
\label{sec:distr-mixed-pack}

In this section, we consider the mixed packing-covering problem
(\MixPC), and prove our main result for general \MixPC in the
distributed setting.
As mentioned earlier, we build on the parallel algorithm of
\cite{Young2001MixedPackingCovering}. We start by formally defining
the problem and the goals; then we outline Young's approach and
highlight the challenges in the \congest model, before proposing our
solution and its analysis.

\subsection{The Objects of Interest}
\label{sec:MPI-defs}

\begin{defn}[Mixed Packing-Covering Instance]
  An instance $\calI$ of a (feasibility) mixed packing-covering (\MixPC)
  linear program consists of matrices
  $P, C \in \RR_{\geq 0}^{m\times n}$ and vectors
  $p, c \in \RR_{\geq 0}^m$. The goal is to find a solution to the
  following (feasibility) LP:
  \begin{align}
    Px &\leq p \label{eq:mpc}\\
    Cx &\geq c \notag \\
    x & \geq 0, \notag
  \end{align}
\end{defn} 

\begin{defn}[Distributed \MixPC Instance]
  \label{def:mixPC}
  A distributed mixed packing-covering (\MixPC) instance $\calI$ is 
  given by matrices $P, C \in \RR^{m\times n}_{\geq 0}$ and right-hand
  sides $p, c \in \RR^{m}_{\geq 0}$.
  This instance is distributed over its communication network
  $H(\calI)$, defined as follows: there are $2m+n$ nodes, with a node
  for each (packing or covering) row and for each column. There is an
  edge between the packing row node $i$ and column node $j$ whenever
  $P_{ij} \neq 0$; similarly, there is an edge between the covering
  row node $i$ and column node $j$ whenever $C_{ij}\neq 0$. Each node
  knows only the entries of the $P, C$ matrices corresponding to their
  row or column; each row node $i$ additionally knows the
  corresponding right-hand-side values $p_i$ or $c_i$.
\end{defn}

\begin{defn}[Distributed $(1+\eps)$-Approximation for \MixPC]
  A distributed protocol for \MixPC is a $(1+\eps)$-approximation if, for
  any distributed \MixPC instance $\calI$ (Defn.~\ref{def:mixPC})
  (where every node is also given an approximation parameter $\eps >
  0$), it terminates in a finite number of communication rounds. Upon
  termination,
  \begin{itemize}[nosep]
  \item every column node $j$ outputs a value $x_j \geq 0$, 
  \item every covering row node $i$ outputs \textsc{accept} if $(Cx)_i \geq c_i$ (and \textsc{reject} otherwise), and 
  \item every packing row node $i$ outputs \textsc{accept} if $(Px)_i \leq (1+\eps)p_i$ (and \textsc{reject} otherwise). 
  \end{itemize}
  Moreover, if any node outputs \textsc{reject}, then the instance $\calI$ is infeasible.
\end{defn}  

Our main result is the following:

\begin{restatable}[Main \MixPC Theorem]{theorem}{distyoung}
  \label{thm:dist-young}
  There is a deterministic distributed protocol for the \MixPC
  problem, such that if each node is given a parameter $\eps > 0$, it
  is guaranteed to terminate in $\poly(\log(mn)/\eps)$ rounds
  in the \congest model, and return a
  $(1+\eps)$-approximation. Moreover, the protocol ensures that packing nodes never return
\textsc{reject}. Equivalently, for the vector $x$ returned by the
protocol, the packing-side guarantee
\[
  \max_i (Px)_i \le (1+O(\eps))L
\]
holds regardless of whether all covering nodes output \textsc{accept}
or some covering node outputs \textsc{reject}.
\end{restatable}

\subsection{Overview of Young's Parallel Algorithm} 
\label{sec:youngs-algorithm}

The algorithm of \cite{Young2001MixedPackingCovering} starts by
rescaling the mixed packing-covering LP to the following form, where
we want to find $x \in \RR^n$ satisfying:
\begin{align}
  Px &\leq L, \tag{MPC} \label{eq:mcp_standard}\\
  Cx &\geq L, \notag\\
   x &\geq 0. \notag
\end{align}
where $L$ is $\Omega(\log m /\varepsilon)$. (This rescaling can be done in one round of
the \congest model, and henceforth we assume this rescaled version
without loss of generality.) Young's approach is to view the covering
and packing constraints compactly as asking for $\max(Px) \leq L$ and
$\min(Cx) \geq L$, and then to use the smoother functions $\softmax$
and $\softmin$ to approximate the $\max$ and $\min$ functions.  Here
we define $\softmax(y) := \ln (\sum_i e^{y_i})$ and
$\softmin(y) := - \softmax (-y) = -\ln(\sum_ie^{-y_i})$ for vectors
$y \in \RR^m$. Apart from being differentiable, these functions have
the following properties:
\begin{gather}
  \softmax(y) \geq \max(y) \geq \softmax(y) - \ln m ~~~\text{and}~~~ 
  \softmin(y) \leq \min(y) \leq \softmin(y) + \ln m; \label{eq:1}
\end{gather}
choosing $L = \Omega(\log m /\varepsilon)$ ensures that the smooth
approximations are within a $(1 + \varepsilon)$ factor of the original
$\min$/$\max$ functions, when the max/min values are themselves
$\Theta(L)$.

The partial derivatives of the functions are $\partial_{y_i}
\softmax(y) = \nicefrac{e^{y_i}}{\sum_{k}e^{y_{k}}}$ and $\partial_{y_i}
\softmin(y) = \nicefrac{e^{-y_i}}{\sum_{k} e^{-y_{k}}}$. Young shows the following lemma:
\begin{lemma}
    \label{lem:young_key}
    If \eqref{eq:mcp_standard} is feasible, then for any $x \geq 0$,
    there exists $j$ such that
    \[ \partial_j\softmax(Px) \leq \partial_j
    \softmin(Cx).\footnote{For brevity, henceforth we use
      $\partial_j f(x)$ to denote $\partial_{x_j} f(x)$.} \]
\end{lemma}
The algorithm proceeds in \emph{rounds}. In every round, it finds some
set of variables $x_j$ (whose existence is ensured by
\Cref{lem:young_key}) such that
$\partial_j \softmax(Px) \leq (1+\varepsilon) \partial_j
{\softmin(Cx)}$ and increases them by a small amount, so that the increment
of both $\softmax(Px)$ and $\softmin(Cx)$ is at most $\varepsilon$.
By doing this, it ensures that
$\softmax(Px) \leq (1 + O(\varepsilon))\softmin(Cx)$ throughout the
process, and hence we terminate at a solution $x$ such that
$\softmin(Cx) = (1 + \Theta(\varepsilon)) L$ and
$\softmax(Px) \leq (1 + O(\varepsilon))L$. Since the max/min functions
and their smoothed variants differ by only a logarithmic
term~(\ref{eq:1}), setting the value of $L = \Omega(\log m/\eps)$
ensures that the violation is at most $\eps L$, giving us
$(1 + \varepsilon)$-feasibility of $x$.

\newcommand{\infrac}[2]{\big(#1\big)/\big(#2\big)}

\begin{algorithm}[ht]
  \caption{ParallelYoung}
  \label{alg:young-parallel}
  \KwIn{Matrices $P, C$, scalar $\varepsilon$.}
  \KwOut{``infeasible'' or $x \geq 0$ s.t.\ $Px \leq (1+O(\varepsilon))\,L$ and $Cx \geq L$.}
  $x_j \gets \min_i \nicefrac{1}{(nP_{ij})}$ for each $j$ \;
  $L \gets \bigl(4\ln m\bigr)/\varepsilon$ \;
  \label{line:local} define $\loc_j(x) \coloneqq \infrac{\sum_i P_{ij}\, e^{(Px)_i}}{\sum_i C_{ij}\, e^{-(Cx)_i}}$.\;
  \label{line:global} define $\glo(x) \coloneqq \infrac{\sum_i e^{(Px)_i}}{\sum_i e^{-(Cx)_i}}$.\;
  \While{$\min\, Cx < L$}{
    \label{line:check-G} \If{$\tau$ is not yet set, or $\min_j \loc_j(x)/\tau > 1 + \varepsilon$}{
      \label{line:compute-G} $\tau \gets \glo(x)$ \tcp*{Recompute $\tau$, start a new phase}
      \lIf{$\min_j \loc_j(x)/\tau > 1 + \varepsilon$}{\Return ``infeasible''}
    }
    Delete the $i$-th covering constraint of $C$ for each $i$ s.t.\ $(Cx)_i \geq L$\;
    \ForEach{$j$}{
      set $\bm{\alpha}_j \gets (\varepsilon/10) \cdot x_j /L $ if  $\loc_j(x)/\tau
      \leq 1+\varepsilon$, else $\bm{\alpha}_j \gets 0$
    }
    $x \gets x + \bm{\alpha}$\; \label{line:update_same_time}
  }
  \Return $x$\;
\end{algorithm}









 
While this basic idea already ensures convergence, the number of
parallel rounds can be high, depending on how we choose the set of
variables to update on each round. To ensure that the algorithm runs
in $\poly(\log(m)/\varepsilon)$ parallel rounds, Young splits the
computation of
$\ratio_j(x) := \frac{\partial_j\softmax(Px)}{\partial_j  \softmin(Cx)}$ into two components: define
$\ratio_j(x) = \loc_j(x) /\glo(x)$, where $\loc_j(x)$ captures the
terms that depend only on $j$. (See \Cref{alg:young-parallel}, and in
particular, lines~\ref{line:local} and \ref{line:global}.)

Now, as $x$ changes, we only recompute $\glo(x)$ occasionally---at the
start of each \emph{phase}---and use $\tau$ to denote the value of the
last computation of $\glo(x).$ Since both $(Px)_i$ and $(Cx)_i$ are
non-decreasing, the value of $\glo(x)$ is non-decreasing. As a result,
$\ratio_j(x)$ can be safely approximated by $\loc_j(x)/\tau$, allowing
the updated variables to continue satisfying
$\partial_j {\softmax(Px)} \leq (1+\varepsilon)\; \partial_j
{\softmin(Cx)}$ throughout the phase without needing a global
recomputation.

Note that each phase (where we update the value of $\tau$) consists of
many rounds, in which we raise the values of some subset of variables
in parallel. The following lemmas bound the total number of parallel
rounds; each round requires $\text{nnz}(P) + \text{nnz}(C)$ amounts of work.

\begin{lemma}
  \label{lem: phase_increment_g}
  If \eqref{eq:mcp_standard} is feasible, $\tau$ increases by at least a
  factor of $(1 + \varepsilon)$ in each phase. Moreover, there are
  $O(\nicefrac{L}{\varepsilon})$ phases in total.
\end{lemma}
\begin{proof}
  We recompute $\glo(x)$ only when $\loc_j(x)/\tau > (1 + \varepsilon)$
  for all $j$. By \Cref{lem:young_key}, if \eqref{eq:mcp_standard} is
  feasible, there always exists an entry $j$ such that
  $\ratio_j(x) = \loc_j(x)/\glo(x) \leq 1$. This means that
  $\glo(x) > (1 + \varepsilon)\tau$, and hence the new value of $\tau$
  increases by at least a factor of $(1 + \varepsilon).$

  To bound the total number of phases, note that the initial value of
  $\tau$ is at least $1$, because $x \geq 0$ implies that $(Px)_i \geq 0$
  and $(Cx)_i \geq 0$, making the numerator of $\glo(x)$ at least
  $\sum_i e^0 = m$, and the denominator at most $\sum_i e^0 = m$. The
  final value is $e^{O(L)}$ because the algorithm terminates when
  $\min(Cx) \ge L$. Throughout the updates,
  $\softmax(Px) \le (1+O(\eps))\softmin(Cx)$, meaning both max packing
  and min covering are scaled to roughly $O(L)$. Therefore,
  $\glo(x) = e^{\softmax(Px) + \softmin(Cx)}$ will not exceed
  $e^{2L + O(\eps L)}$, which is at most $e^{O(L)}$.
\end{proof}


The following claims follow the argument in the proof of 
\cite[Lemma~7]{Young2001MixedPackingCovering}, restated in our
notation.

\begin{lemma}
  \label{lem:bounding_softminmax}
  Throughout any execution in which every update satisfies
  $\loc_j(x)/\glo(x)\leq 1+\varepsilon$, we have
  \[
    \softmax(Px)
    \leq
    (1+O(\varepsilon))\softmin(Cx)+O(\log m).
  \]
\end{lemma}

\begin{proof}
  This is the standard smooth-potential invariant in Young's analysis:
  each such update satisfies
  \(
    \Delta\softmax(Px)
    \leq
    (1+O(\varepsilon))\Delta\softmin(Cx).
  \)
  Summing over the updates, together with the initialization
  $\softmax(Px)\leq 1+\ln m$ and $\softmin(Cx)\geq -\ln m$, gives the
  claim.
\end{proof}

\begin{claim}
  \label{clm:young-bounded-packing-load}
  At the beginning of any round in which some variable is updated,
  $\max_i (Px)_i \leq (1 + O(\varepsilon))L$.
\end{claim}

\begin{proof}
  At the beginning of such a round, the algorithm has not terminated,
  so $\softmin(Cx) \leq \min_i (Cx)_i < L$. By
  \Cref{lem:bounding_softminmax},
  \[
    \softmax(Px) \leq (1+O(\varepsilon))\softmin(Cx) + O(\log m).
  \]
  Since $L=\Theta(\log m/\varepsilon)$, it follows that
  $\softmax(Px) \leq (1 + O(\varepsilon))L$. The claim follows from
  $\max_i (Px)_i \leq \softmax(Px)$.
\end{proof}

\begin{lemma}
  \label{lem:young_phase_updates}
  Each variable $x_j$ is updated in at most
  $\widetilde{O}(\log m \log n/\varepsilon^2)$ rounds over the entire run
  of the algorithm.
\end{lemma}
\begin{proof}
  Fix a variable $x_j$, and let $P_{\max,j}:=\max_i
  P_{ij}$. Initially, \( x_j^{(0)}=\frac{1}{nP_{\max,j}} .  \)
  Moreover, each update of $x_j$ multiplies it by
  $1+\varepsilon/(10L)$.  Let $T_j$ be the total number of updates to
  $x_j$. If $T_j=0$, there is nothing to prove. Otherwise, just before
  the last update of $x_j$, by \Cref{clm:young-bounded-packing-load},
  \[
    P_{\max,j}x_j \leq \max_i (Px)_i \leq (1+O(\varepsilon))L = O(L).
  \]
  At this point, $ x_j = (1+\varepsilon/(10L))^{T_j-1} \; x_j^{(0)}, $
  and hence
  \[
    (1+\varepsilon/(10L))^{T_j-1} \leq O(nL).
  \]
  Taking logarithms gives
  \[
    T_j = O((L/\varepsilon)\log(nL))
      = \widetilde{O}((\log m\log n)/\varepsilon^2),
  \]
  since $L=\Theta(\log m/\varepsilon)$.
\end{proof}

\begin{lemma}
  \label{lem:young_prefix_update}
  For any variable \(x_j\), the rounds in which \(x_j\) is updated
  form a prefix of the parallel rounds in any phase.
\end{lemma}

\begin{proof}
  Within a phase, \(\loc_j(x)\) is non-decreasing while \(\tau\)
  remains fixed. The statement then follows from the fact that
  $\loc_j/\tau$ is non-decreasing.
\end{proof}

These lemmas imply that each phase consists of at most
\(\widetilde{O}(\log m \log n / \varepsilon^2)\) parallel rounds,
since each round within a phase has at least one variable being
updated.  Moreover, there are $O(L/\varepsilon)$ phases, so we get a
$(1 + O(\varepsilon))$-feasible solution in
$\widetilde{O}(\log^2 m \log n/\varepsilon^4)$ parallel rounds.

\subsection{Extension to the \congest Model}
\label{sec:young-distr}

Extending \Cref{alg:young-parallel} to the \congest model has one
central barrier: the agents need the value of $\tau$ to make progress,
and communicating this globally would require
$\Omega(\text{diameter})$ rounds in general. However, there is a
relatively clean solution: we perform a doubling search on the value
$\tau$. Specifically, there are two changes:
\begin{enumerate}[nosep]
\item Instead of the recomputation of $\tau$ in line~\ref{line:compute-G}
  of \Cref{alg:young-parallel}, we iterate over values of $\tau$ taking
  values equal to powers of $(1 + \varepsilon)$, ranging from $1$ to
  $\tau_{\max}$.
\item Within each phase, we remove the check for whether the value of
  $\tau$ is a good approximation for $\glo$ in line~\ref{line:check-G} of
  \Cref{alg:young-parallel}, and instead just perform the variable
  updates for $\widetilde{O}(\log m \log n/\varepsilon^2)$ rounds of
  \congest.
\end{enumerate}
The formal algorithm appears as \Cref{alg:mpc-distributed}.


\begin{algorithm}[H]
  \caption{DistributedYoung}
  \label{alg:mpc-distributed}
\KwIn{A distributed \MixPC instance $\calI$ as in \Cref{def:mixPC},
rescaled to the standard form \eqref{eq:mcp_standard}, and an accuracy
parameter $\varepsilon$.}
  \KwOut{ For each variable $j$, output $x_j$;
    each covering constraint outputs \textsc{accept}/\textsc{reject}.}
    \ForEach{node in the communication graph}{
          $\varepsilon' \gets \varepsilon/c$ \tcp*{$c$ large enough constant to get $(1 + \varepsilon)$ feasible solution}
$L \gets \bigl(4\ln m\bigr)/\varepsilon'$ \;
  $\tau \gets 1, \tau_{\max} \gets e^{10L}, T\gets  \left\lceil \log_{1+\varepsilon'}(\tau_{\max}/\tau)\right\rceil$ \;
   $R\gets \widetilde{O}(\log m\log n/\eps'^2)$ chosen larger than the
  bound in \Cref{lem:young_phase_updates}.
  
          \tcp{Every node computes the same $\varepsilon', L, \tau, T$ from the common inputs $m,\varepsilon$;
          }
    }
  \lForEach{variable $j$}{$x_j \gets \min_i \nicefrac{1}{(nP_{ij})}$}
  \lForEach{covering constraint $i$}{$delete_i \gets \False$}
  \For{$t \gets 1$ \KwTo $T$}{
    \tcp{Start of a phase}
    \For{$r \gets 1$ \KwTo $R$ }
    { 

      \ForEach{\text{covering constraint}  $i$ with $delete_i = \False$}{
        $v \gets (Cx)_i$ \tcp*{1 round} 
        \lIf{$v \geq L$}{$delete_i \gets \True$}
      }\tcp{Deleted covering constraint will not participate in future computation.}
      \ForEach{variable $j$}{
        $q_j \gets \texttt{ComputeLocal}(j, x)$ \tcp*{recompute $\loc_j$ at the current $x$; see Procedure~\ref{proc:compute-local}}
        \lIf{$q_j/\tau \leq 1+\varepsilon'$}{$\bm{\alpha}_j \gets (\varepsilon'/10)\,x_j/L$}
        \lElse{$\bm{\alpha}_j \gets 0$}
        $x_j \gets x_j + \bm{\alpha}_j$
      }
    }
    All nodes compute: $\tau \gets (1+\varepsilon')\tau$
  }
  \lForEach{variable node \(j\)}{output \(x_j\)} 
  \lForEach{packing constraint \(i\)}{output \textsc{accept}}
  \tcp{For every packing constraint $i$, we ensure that $(Px)_i \leq (1+\varepsilon)L$.}
  \lForEach{covering constraint \(i\)}{output \textsc{accept} if
  $(Cx)_i \geq L$,  else \textsc{reject}}
\end{algorithm}
\begin{procedure}[H]
\caption{ComputeLocal($j$, $x$)}
\label{proc:compute-local}
\KwOut{the value $\loc_j(x)$ at the current iterate $x$.}
\Return $\displaystyle
  \frac{\sum_{i \in N(j)} P_{ij}\, e^{(Px)_i}}
       {\sum_{i\in N(j):\, delete_i = \False} C_{ij}\, e^{-(Cx)_i}}$
  \tcp*{variable $j$ gathers from its incident constraints in $1$ \congest rounds; deleted covering constraints are excluded from the denominator}
\end{procedure}

\subsubsection{Correctness of \Cref{alg:mpc-distributed}}
\label{sec:analysis-distrib-MPC}


For simplicity, in this subsection we write $\eps$ for the internal
accuracy parameter $\varepsilon'$ used in \Cref{alg:mpc-distributed}.
All quantities involving covering constraints are computed only over
the currently active covering rows.

We call an update of a variable $x_j$ \emph{safe} if, immediately
before the update,
\[
  \frac{\loc_j(x)}{\glo(x)} \leq 1+\eps .
\]

\begin{lemma}
  \label{lem:clean_cor}
  Suppose \eqref{eq:mcp_standard} is feasible. In
  \Cref{alg:mpc-distributed}, at the beginning of every phase, either
  all covering constraints are already met or
  \[
    \tau \leq \glo(x).
  \]
  Moreover, every update performed by the algorithm is safe, and at the
  end of every phase, either all covering constraints are met or
  \[
    \glo(x) > (1+\eps)\tau,
  \]
  where $\tau$ denotes the value used during that phase.
\end{lemma}

\begin{proof}

  We prove by induction over phases that, at the beginning of
  every phase, either all covering constraints are already met or
  $\tau\leq\glo(x)$. Initially, $\tau=1$. Since $x\geq 0$, the numerator
  of $\glo(x)$ is at least $m$, and the denominator is at most $m$.
  Hence $\glo(x)\geq 1=\tau$.

  Now using induction assumption, consider a phase whose starting value satisfies
  $\tau\leq\glo(x)$. If all covering constraints are already met, the statement is true. We focus on the non-trivial case that not all covering constraints are met. Inside a phase, we may do two kinds of operations: increase the value of $x$ or delete some covering row. In either case, for $\glo(x)$, the numerator is non-decreasing, while the
  denominator is non-increasing. Thus $\glo(x)$ is non-decreasing and $\tau\leq\glo(x)$ throughout the
  phase. Hence whenever the algorithm updates a variable $x_j$ in this
  phase,
  \[
    \frac{\loc_j(x)}{\glo(x)}
    \leq
    \frac{\loc_j(x)}{\tau}
    \leq
    1+\eps .
  \]
  Therefore every update in this phase is safe.

  Suppose that, at the end of the phase, not all covering constraints
  are met. We claim that no variable is eligible at this point. Indeed,
  if some $x_j$ satisfied
  \[
    \frac{\loc_j(x)}{\tau}\leq 1+\eps
  \]
  at the end of the phase, then by \Cref{lem:young_prefix_update}, the
  same variable was eligible in every inner round of this phase. Hence
  $x_j$ would have been updated in all $R$ inner rounds, contradicting
  \Cref{lem:young_phase_updates} and the choice of $R$.

  Therefore,
  \[
    \loc_j(x) > (1+\eps)\tau
    \qquad \text{for every } j.
  \]
  Since the active instance remains feasible, \Cref{lem:young_key}
  gives a variable $x_{j^*}$ such that
  \[
    \frac{\loc_{j^*}(x)}{\glo(x)}\leq 1.
  \]
  Combining the two inequalities gives
  \[
    (1+\eps)\tau < \loc_{j^*}(x) \leq \glo(x).
  \]
  Thus $\glo(x)>(1+\eps)\tau$ at the end of the phase. After the
  update $\tau\gets (1+\eps)\tau$, the invariant
  $\tau\leq\glo(x)$ holds for the next phase.
\end{proof}

\begin{lemma}
  \label{lem:distributed-mpc-correct}
  In every execution of \Cref{alg:mpc-distributed}, the returned vector
  $x$ satisfies the packing-side bound
  \[
    \max_i (Px)_i \leq (1+O(\eps))L.
  \]
  Moreover, if \eqref{eq:mcp_standard} is feasible, then all covering
  constraint nodes output \textsc{accept}.
\end{lemma}
\begin{proof}
  We first prove the packing-side guarantee. This part does not use
  feasibility of \eqref{eq:mcp_standard}, and hence remains true even
  if some covering constraint node eventually outputs \textsc{reject}.

  If no variable is ever updated, then the initial choice of $x$ gives
  $\max_i(Px)_i\le 1\le L$, and there is nothing to prove. Otherwise,
  consider the final parallel update round of the algorithm. Let $x^-$
  and $x^+$ denote the vectors immediately before and immediately after
  this update round, respectively.

  Immediately before this update round, some variable is updated. Hence
  by \Cref{clm:young-bounded-packing-load},
  \[
    \max_i(Px^-)_i \leq (1+O(\eps))L.
  \]
  During this update round, each variable is either unchanged or
  multiplied by $1+\eps/(10L)$. Since all entries of $P$ are
  nonnegative, for every packing row $i$ we have
  \[
    (Px^+)_i
    \leq
    \left(1+\frac{\eps}{10L}\right)(Px^-)_i
    \leq
    (1+O(\eps))L.
  \]
  After this final update round, the algorithm may still delete
  covering rows or increase $\tau$, but it never changes $x$ again.
  Therefore the final returned vector satisfies
  \[
    \max_i(Px)_i \leq (1+O(\eps))L.
  \]

  It remains to prove the covering-side statement, so from now on
  assume that \eqref{eq:mcp_standard} is feasible. By
  \Cref{lem:clean_cor}, every update made by the algorithm is safe, and
  hence \Cref{lem:bounding_softminmax} applies throughout the execution.

  Let $\tau_{\max}=e^{10L}$, as in \Cref{alg:mpc-distributed}, and let
  $\bar\tau$ be the value of $\tau$ used during the final phase. Since
  \[
    T=\left\lceil \log_{1+\eps}(\tau_{\max}/\tau)\right\rceil,
  \]
  the value of $\tau$ after the final phase update is
  $(1+\eps)\bar\tau\ge \tau_{\max}$. Moreover, this value overshoots
  $\tau_{\max}$ by a factor at most $1+\eps$.

  Suppose, for contradiction, that some covering constraint node outputs
  \textsc{reject}. Then not all covering constraints are met at the end
  of the final phase. By \Cref{lem:clean_cor}, applied to the final
  phase,
  \[
    \glo(x) > (1+\eps)\bar\tau \geq \tau_{\max}=e^{10L}.
  \]
  On the other hand,
  \[
    \glo(x)
    =
    \frac{\sum_i e^{(Px)_i}}{\sum_i e^{-(Cx)_i}}
    =
    e^{\softmax(Px)+\softmin(Cx)}.
  \]
  By \Cref{lem:bounding_softminmax},
  \[
    \log \glo(x)
    \leq
    (2+O(\eps))\softmin(Cx)+O(\log m).
  \]
  Therefore,
  \[
    10L
    <
    (2+O(\eps))\softmin(Cx)+O(\log m).
  \]
  Since $L=\Theta(\log m/\eps)$ and $\eps$ is sufficiently small, this
  implies
  \[
    \softmin(Cx)>3L.
  \]

  The value $\softmin(Cx)$ is computed over the active covering rows.
  Moreover, any covering row that outputs \textsc{reject} must still be
  active: once a covering row is deleted, its value is already at least
  $L$, and this value can only increase afterwards. Hence the active
  covering set is nonempty. Since
  \[
    \softmin(Cx)\leq \min_{i:\, i\text{ active}} (Cx)_i,
  \]
  every active covering row has value greater than $3L$. This
  contradicts the existence of an active covering row with value less
  than $L$.  Thus no covering constraint node outputs
  \textsc{reject}. Equivalently, every covering constraint node
  outputs \textsc{accept}.
\end{proof}

\subsubsection{Distributed Implementation of \Cref{alg:mpc-distributed}}
\label{sec:distr-imple-young}

\Cref{alg:mpc-distributed} can be implemented in the \congest model. Every node
maintains the same phase parameter $\tau$ and the same inner-round counter
$r$ locally. Since all nodes start with the same initial value of $\tau$, and
since the number of inner rounds in each phase is fixed in advance, all
nodes enter and leave every phase simultaneously. 
Moreover, the computation of $\loc$ functions can be done in one round of computation with $O(\log n)$ bits of communication. Therefore, the entire procedure can be executed synchronously
in the \congest model, with each parallel round of the algorithm simulated
by a constant number of \congest rounds up to the standard rounding needed
to represent numerical values using $O(\log n)$ bits.

The final accept/reject outputs are interpreted under the local-output
convention defined earlier.

\begin{lemma}
    \label{lem:distributed-mpc-polylog}
    \Cref{alg:mpc-distributed} works in $\widetilde{O}(\log^2m \log n/\varepsilon^4)$ number of \congest rounds.
\end{lemma}
\begin{proof}
  The total number of rounds depends on the number of phases and the
  number of iterations per phase. The phase multiplier $\tau$ starts
  at $1$ and is multiplied by $(1+\varepsilon')$ at the end of each
  phase until it reaches or exceeds $\tau_{\max}=e^{10L}$. Since
  $L=\Omega(\log m/\varepsilon')$, the number of phases is
  \[
    O(L/\varepsilon')=O(\log m/(\varepsilon')^2).
  \]
  Because $\varepsilon'=\Theta(\varepsilon)$, this is
  \(O(\log m/\varepsilon^2)\).
    
  Within each phase, the algorithm executes exactly
  $\widetilde{O}(\log m \log n / \varepsilon^2)$ parallel inner
  rounds. As established, each inner round requires a constant number
  of \congest communications. Thus, the overall round complexity is
  the product of the number of phases and the inner rounds, which
  evaluates to
  $O(\log m / \varepsilon^2) \times \widetilde{O}(\log m \log n /
  \varepsilon^2) = \widetilde{O}(\log^2 m \log n / \varepsilon^4)$.
\end{proof}

\Cref{lem:distributed-mpc-correct} and
\Cref{lem:distributed-mpc-polylog} directly give \Cref{thm:dist-young}.



{\small
\bibliographystyle{alpha}
\bibliography{references}

@book{Peleg2000,
  author    = {David Peleg},
  title     = {Distributed Computing: A Locality-Sensitive Approach},
  publisher = {SIAM},
  year      = {2000},
  doi       = {10.1137/1.9780898719772}
}

@inproceedings{AhmadiKO18,
  author       = {Mohamad Ahmadi and
                  Fabian Kuhn and
                  Rotem Oshman},
  editor       = {Ulrich Schmid and
                  Josef Widder},
  title        = {Distributed Approximate Maximum Matching in the {CONGEST} Model},
  booktitle    = {32nd International Symposium on Distributed Computing, {DISC} 2018,
                  New Orleans, LA, USA, October 15-19, 2018},
  series       = {LIPIcs},
  volume       = {121},
  pages        = {6:1--6:17},
  publisher    = {Schloss Dagstuhl - Leibniz-Zentrum f{\"{u}}r Informatik},
  year         = {2018},
  url          = {https://doi.org/10.4230/LIPIcs.DISC.2018.6},
  doi          = {10.4230/LIPICS.DISC.2018.6},
  bibsource    = {dblp computer science bibliography, https://dblp.org}
}

@article{CzygrinowHSW16,
  author  = {Andrzej Czygrinow and Michal Hanckowiak and Edyta Szymanska and Wojciech Wawrzyniak},
  title   = {On the distributed complexity of the semi-matching problem},
  journal = {J. Comput. Syst. Sci.},
  volume  = {82},
  number  = {8},
  pages   = {1251--1267},
  year    = {2016},
  doi     = {10.1016/j.jcss.2016.05.001}
}

@inproceedings{AssadiBL20,
  author    = {Sepehr Assadi and Aaron Bernstein and Zachary Langley},
  title     = {Improved Bounds for Distributed Load Balancing},
  booktitle = {34th International Symposium on Distributed Computing ({DISC} 2020)},
  series    = {LIPIcs},
  volume    = {179},
  pages     = {1:1--1:15},
  publisher = {Schloss Dagstuhl -- Leibniz-Zentrum f{\"{u}}r Informatik},
  year      = {2020},
  doi       = {10.4230/LIPIcs.DISC.2020.1}
}

@article{HalldorssonKPR18,
  author  = {Magn{\'{u}}s M. Halld{\'{o}}rsson and Sven K{\"{o}}hler and Boaz Patt-Shamir and Dror Rawitz},
  title   = {Distributed backup placement in networks},
  journal = {Distributed Comput.},
  volume  = {31},
  number  = {2},
  pages   = {83--98},
  year    = {2018},
  doi     = {10.1007/s00446-017-0299-x}
}

@inproceedings{BhattacharyaKS23,
  author    = {Sayan Bhattacharya and Peter Kiss and Thatchaphol Saranurak},
  title     = {Dynamic Algorithms for Packing-Covering LPs via Multiplicative Weight Updates},
  booktitle = {Proceedings of the 2023 Annual {ACM-SIAM} Symposium on Discrete Algorithms ({SODA} 2023)},
  pages     = {1--47},
  publisher = {{SIAM}},
  year      = {2023},
  doi       = {10.1137/1.9781611977554.ch1}
}

@article{HarveyLLT06,
  author  = {Nicholas J. A. Harvey and Richard E. Ladner and L{\'{a}}szl{\'{o}} Lov{\'{a}}sz and Tami Tamir},
  title   = {Semi-matchings for Bipartite Graphs and Load Balancing},
  journal = {J. Algorithms},
  volume  = {59},
  number  = {1},
  pages   = {53--78},
  year    = {2006}
}

@article{FakcharoenpholLN14,
  author  = {Jittat Fakcharoenphol and Bundit Laekhanukit and Danupon Nanongkai},
  title   = {Faster Algorithms for Semi-Matching Problems},
  journal = {{ACM} Trans. Algorithms},
  volume  = {10},
  number  = {3},
  pages   = {14:1--14:23},
  year    = {2014}
}

@article{Vos2026,
  author       = {Tijn de Vos and
                  Leo Wennmann and
                  Malte Baumecker and
                  Yannic Maus and
                  Florian Schager},
  title        =  {Distributed {Santa Claus} via global rounding},
  journal      = {CoRR},
  volume       = {abs/2604.27983},
  year         = {2026},
  url          = {https://doi.org/10.48550/arXiv.2604.27983},
  doi          = {10.48550/ARXIV.2604.27983},
  eprinttype   = {arXiv},
  eprint       = {2604.27983},
  bibsource    = {dblp computer science bibliography, https://dblp.org}
}

@article{KonradR16,
  author  = {Christian Konrad and Adi Ros{\'{e}}n},
  title   = {Approximating Semi-matchings in Streaming and in Two-Party Communication},
  journal = {{ACM} Trans. Algorithms},
  volume  = {12},
  number  = {3},
  pages   = {32:1--32:21},
  year    = {2016},
  doi     = {10.1145/2898960}
}

@inproceedings{AssadiBL23,
  author    = {Sepehr Assadi and Aaron Bernstein and Zachary Langley},
  title     = {All-Norm Load Balancing in Graph Streams via the Multiplicative Weights
               Update Method},
  booktitle = {14th Innovations in Theoretical Computer Science Conference, {ITCS} 2023},
  series    = {LIPIcs},
  pages     = {7:1--7:24},
  publisher = {Schloss Dagstuhl - Leibniz-Zentrum f{\"{u}}r Informatik},
  year      = {2023},
  doi       = {10.4230/LIPICS.ITCS.2023.7}
}

@inproceedings{AssadiBLLW25,
  author    = {Sepehr Assadi and Aaron Bernstein and Zachary Langley and Lap Chi Lau and
               Robert Wang},
  title     = {Streaming and Communication Complexity of Load-Balancing via Matching
               Contractors},
  booktitle = {Proceedings of the 2025 Annual {ACM-SIAM} Symposium on Discrete Algorithms,
               {SODA} 2025},
  pages     = {3423--3449},
  publisher = {{SIAM}},
  year      = {2025},
  doi       = {10.1137/1.9781611978322.113}
}

@inproceedings{Ahmadian2021DistributedLB,
  author    = {Ahmadian, Sara and Liu, Allen and Peng, Binghui and Zadimoghaddam, Morteza},
  title     = {Distributed Load Balancing: {A} New Framework and Improved Guarantees},
  booktitle = {12th Innovations in Theoretical Computer Science Conference ({ITCS} 2021)},
  series    = {Leibniz International Proceedings in Informatics ({LIPIcs})},
  volume    = {185},
  pages     = {79:1--79:22},
  publisher = {Schloss Dagstuhl -- Leibniz-Zentrum f{\"u}r Informatik},
  year      = {2021},
  doi       = {10.4230/LIPIcs.ITCS.2021.79},
  url       = {https://drops.dagstuhl.de/opus/volltexte/2021/13817/},
}

@inproceedings{
argue2022learning,
title={Learning from a Sample in Online Algorithms},
author={C.J. Argue and Alan Frieze and Anupam Gupta and Christopher Seiler},
booktitle={Advances in Neural Information Processing Systems},
editor={Alice H. Oh and Alekh Agarwal and Danielle Belgrave and Kyunghyun Cho},
year={2022},
url={https://openreview.net/forum?id=KMaI40_UaGw}
}

@article{BartalBR04,
  author    = {Yair Bartal and John W. Byers and Danny Raz},
  title     = {Fast, Distributed Approximation Algorithms for Positive Linear Programming with Applications to Flow Control},
  journal   = {SIAM Journal on Computing},
  volume    = {33},
  number    = {6},
  pages     = {1261--1279},
  year      = {2004},
}

@inproceedings{Li23,
  author    = {Shi Li},
  title     = {Nearly-Linear Time {LP} Solvers and Rounding Algorithms for Scheduling Problems},
  booktitle = {50th International Colloquium on Automata, Languages, and Programming ({ICALP} 2023)},
  series    = {Leibniz International Proceedings in Informatics ({LIPIcs})},
  pages     = {86:1--86:20},
  publisher = {Schloss Dagstuhl -- Leibniz-Zentrum f{\"u}r Informatik},
  year      = {2023},
  doi       = {10.4230/LIPIcs.ICALP.2023.86},
  note      = {arXiv:2111.04897},
}

@article{ShmoysT93,
  author    = {David B. Shmoys and {\'E}va Tardos},
  title     = {An Approximation Algorithm for the Generalized Assignment Problem},
  journal   = {Mathematical Programming},
  volume    = {62},
  number    = {1--3},
  pages     = {461--474},
  year      = {1993},
  doi       = {10.1007/BF01585178},
}

@article{LenstraST90,
  author    = {Jan Karel Lenstra and David B. Shmoys and {\'E}va Tardos},
  title     = {Approximation Algorithms for Scheduling Unrelated Parallel Machines},
  journal   = {Mathematical Programming},
  volume    = {46},
  pages     = {259--271},
  year      = {1990},
  doi       = {10.1007/BF01585745},
}

@inproceedings{Young2001MixedPackingCovering,
  author    = {Young, Neal E.},
  title     = {Sequential and Parallel Algorithms for Mixed Packing and Covering},
  booktitle = {Proceedings of the 42nd Annual Symposium on Foundations of Computer Science (FOCS)},
  pages     = {538--546},
  year      = {2001},
  publisher = {IEEE Computer Society},
  doi       = {10.1109/SFCS.2001.959930},
  url       = {https://arxiv.org/abs/cs/0205039}
}

@inproceedings{Mahoney2016ParallelMixedPackingCovering,
  author    = {Mahoney, Michael W. and Rao, Satish and Wang, Di and Zhang, Peng},
  title     = {Approximating the Solution to Mixed Packing and Covering {LPs} in Parallel $\widetilde{O}(\varepsilon^{-3})$ Time},
  booktitle = {Proceedings of the 43rd International Colloquium on Automata, Languages, and Programming (ICALP)},
  pages     = {52:1--52:14},
  year      = {2016},
  publisher = {Schloss Dagstuhl -- Leibniz-Zentrum f{\"u}r Informatik},
  doi       = {10.4230/LIPIcs.ICALP.2016.52},
  url       = {https://drops.dagstuhl.de/entities/document/10.4230/LIPIcs.ICALP.2016.52}
}

@inproceedings{Wang2016DiameterReduction,
  author    = {Wang, Di and Rao, Satish and Mahoney, Michael W.},
  title     = {Unified Acceleration Method for Packing and Covering Problems via Diameter Reduction},
  booktitle = {Proceedings of the 43rd International Colloquium on Automata, Languages, and Programming (ICALP)},
  pages     = {50:1--50:13},
  year      = {2016},
  publisher = {Schloss Dagstuhl -- Leibniz-Zentrum f{\"u}r Informatik},
  doi       = {10.4230/LIPIcs.ICALP.2016.50},
  url       = {https://arxiv.org/abs/1508.02439}
}

@article{AzarNR95,
  author       = {Yossi Azar and
                  Joseph Naor and
                  Raphael Rom},
  title        = {The Competitiveness of On-Line Assignments},
  journal      = {J. Algorithms},
  volume       = {18},
  number       = {2},
  pages        = {221--237},
  year         = {1995},
  url          = {https://doi.org/10.1006/jagm.1995.1008},
  doi          = {10.1006/JAGM.1995.1008},
  bibsource    = {dblp computer science bibliography, https://dblp.org}
}

@article{AspnesAFPW97,
  author       = {James Aspnes and
                  Yossi Azar and
                  Amos Fiat and
                  Serge A. Plotkin and
                  Orli Waarts},
  title        = {On-line routing of virtual circuits with applications to load balancing
                  and machine scheduling},
  journal      = {J. {ACM}},
  volume       = {44},
  number       = {3},
  pages        = {486--504},
  year         = {1997},
  url          = {https://doi.org/10.1145/258128.258201},
  doi          = {10.1145/258128.258201},
  bibsource    = {dblp computer science bibliography, https://dblp.org}
}

@book{Peleg00,
  author    = {David Peleg},
  title     = {Distributed Computing: A Locality-Sensitive Approach},
  publisher = {SIAM},
  year      = {2000},
  doi       = {10.1137/1.9780898719772}
}

@inproceedings{GM26-soda,
  author    = {Anupam Gupta and
               Marco Molinaro},
  title     = {Learning Packing and Covering from Samples},
  booktitle = {SODA},
  year      = {2026},
  month     = {Jan},
  anupamref = {young-sample},
  anupamyear = 2026,
  url = {https://epubs.siam.org/doi/epdf/10.1137/1.9781611978971.48},
}

@inproceedings{LX21,
  author    = {Shi Li and
               Jiayi Xian},
  editor    = {Marina Meila and
               Tong Zhang},
  title     = {Online Unrelated Machine Load Balancing with Predictions Revisited},
  booktitle = {Proceedings of the 38th International Conference on Machine Learning,
               {ICML} 2021, 18-24 July 2021, Virtual Event},
  series    = {Proceedings of Machine Learning Research},
  volume    = {139},
  pages     = {6523--6532},
  publisher = {{PMLR}},
  year      = {2021},
  url       = {http://proceedings.mlr.press/v139/li21w.html},
  bibsource = {dblp computer science bibliography, https://dblp.org}
}

@inproceedings{im2024online,
  title={Online load and graph balancing for random order inputs},
  author={Im, Sungjin and Kumar, Ravi and Li, Shi and Petety, Aditya and Purohit, Manish},
  booktitle={Proceedings of the 36th ACM Symposium on Parallelism in Algorithms and Architectures},
  pages={491--497},
  year={2024}
}

@INPROCEEDINGS{GKS14-matching,
  author = {Anupam Gupta and Amit Kumar and Cliff Stein},
  title = {Maintaining Assignments Online: Matching, Scheduling, and Flows},
  booktitle = {SODA},
  year = {2014},
  pages = {468-479},
  anupamref = {dynamic-match},
  anupamyear = {2014},
  doi = {http://dx.doi.org/10.1137/1.9781611973402.35}
}

@inproceedings{krishnaswamy2023online,
  title={Online unrelated-machine load balancing and generalized flow with recourse},
  author={Krishnaswamy, Ravishankar and Li, Shi and Suriyanarayana, Varun},
  booktitle={Proceedings of the 55th Annual ACM Symposium on Theory of Computing},
  pages={775--788},
  year={2023}
}
}

\appendix
\section{Details of Rounding Algorithm}
\label{app:rounding}

We give the details of how the work of Li~\cite{Li23} can be extended to the \congest setting to round fractional load-balancing solutions to integer ones. Recall the main result claimed in \S\ref{sec:rounding}.

\rounding*

We will use the following distributed representation:
\begin{defn}[Locally represented virtual graph]
  \label{def:virtual-graph}
  A virtual bipartite graph $H=(U\cup J, E_H)$ is \emph{locally
    represented} over $G$ if each node $u\in U$ is \emph{hosted} by
  some machine $i\in M$ (which knows the set of nodes it hosts), and
  each edge $(u,j)\in E_H$ with $u$ hosted by $i$ is realized over a
  $G$-edge $(i,j)$, so that the two endpoints can exchange messages in
  $O(1)$ \congest rounds.
\end{defn}

\subsection{The Rounding Algorithm}
\label{sec:rounding-algorithm}

\begin{algorithm}[H]
  \caption{$(2+O(\eps))$-rounding (\cite{Li23}, simulated in \congest)}
  \label{alg:rounding}
  \DontPrintSemicolon \KwIn{ Each node knows $\varepsilon$. Each machine
    node $i$ knows $s_{ij} \; \forall j \in N(i)$ . Each job node $j$
    knows $s_{ij} \; \forall i \in N(j)$. Each machine node $i$ knows
    $\{ x_{ij} \mid j \in N(i)$ in:
    \[ \textstyle \{ x \mid \sum_{i\in N(j)} x_{ij}\ge 1 \;\forall j, \; \sum_{j\in
        N(i)} s_{ij}x_{ij}\le(1+O(\eps))\lambda^* \;\forall i, \;
      x_{ij}=0 \text{ if } s_{ij}>(1+\eps)\lambda^*\}. \]}
  \KwOut{Each machine $i$ outputs a set $S_i \subseteq N(i)$ satisfying the
    guarantees of \Cref{thm:rounding}.}

  \ForEach{machine $i$ \emph{in parallel}}{
    \tcp{All computation except adding $H$-edges is local to machine $i$; we write it as a centralized procedure}
    $d_i \gets |N(i)|$\;
    Sort all jobs in $N(i)$ by \emph{decreasing} $s_{ij}$; let the result be $j_1,\dots,j_{d_i}$\;
    $\mu \gets 0$;\quad $r \gets 1$; \quad $\gamma \gets \frac{1}{1 + \eps}$
    \tcp*{$\mu$: mass in current group; $r$: current group index}
    \For{$t \gets 1$ \KwTo $d_i$}{
      $v \gets x_{i, j_t}$ \tcp*{share of $j_t$ not yet placed}
      \While{$v > 0$}{
        \lIf{$\mu = \gamma$}{
          $r \gets r+1$;\quad $\mu \gets 0$
        }
        $\delta \gets \min\{v,\ \gamma-\mu\}$\;
        add edge $(i_r,\,j_t)$ to $E_H$\;
        $\mu \gets \mu + \delta$;\quad $v \gets v - \delta$\;
      }
    }
    $k_i \gets r$
    \tcp*{machine $i$ hosts $k_i=\lceil(1+\eps)\sum_j x_{ij}\rceil$ groups, each of mass $\le\gamma$}
  }
  $U \gets \bigcup_i \{i_1,\dots,i_{k_i}\}$
  \tcp*{vertices of virtual graph $H=(U\cup J,E_H)$.
  }
  $H\gets U\cup J$
  
  $\mathcal{M} \gets \texttt{DistributedMatch}(H, E_H, \varepsilon)$
  \tcp*{
    uses $H$'s $(1+\eps)$-expansion, see \Cref{lem:distributed-match}}
  \ForEach{machine $i$}{
    \(S_i \gets \{j : (i_r,j)\in \mathcal{M} \text{ for some } 1\le r\le k_i\}\);\quad output $S_i$\;
  }
\end{algorithm}

\begin{lemma}[\texttt{DistributedMatch}]
  \label{lem:distributed-match}
  Let $H=(U\cup J, E_H)$ be a virtual bipartite graph, locally represented
  over $G$ as in \Cref{def:virtual-graph}, where $U=\bigcup_{i\in \mathcal{M}}\{i_1,
  \dots,i_{k_i}\}$ is the set of groups and each group $i_r$ is hosted by
  machine $i$ (which knows its own $k_i$). Suppose $H$ satisfies:
  \begin{enumerate}[nosep]
  \item \textnormal{(bounded congestion)} every job $j\in J$ is incident
    to at most $O(1)$ groups hosted by any single machine $i$, so each
    underlying $G$-edge $(i,j)$ carries only $O(1)$ edges of $H$; and
  \item \textnormal{($(1+\eps)$-expansion)} $|N_H(J')|\ge(1+\eps)|J'|$ for
    every $J'\subseteq J$.
\end{enumerate}
Then there is a deterministic \congest algorithm
$\texttt{DistributedMatch}(U, E_H, \varepsilon)$ that augments a matching
$\mathcal{M}$ along shortest augmenting paths in $H$ until none remains, in
$\poly((\log nm)/\varepsilon)$ rounds, at the end of which $\mathcal{M}\subseteq
E_H$ matches every job $j\in J$ to exactly one group $i_r\in U$,
represented locally as in \Cref{def:virtual-graph}: each matched edge
$(i_r,j)\in \mathcal{M}$ is known to both its host machine $i$ and job $j$.
\end{lemma}

\paragraph{Correctness.}
The rounding procedure in \Cref{alg:rounding} is exactly Li's virtual-graph
rounding algorithm~\cite{Li23}, interpreted on the virtual graph $H$ constructed
above. We use the following guarantee from~\cite{Li23}: given a
$(1+\eps)$-approximate fractional solution, the grouping construction produces
a virtual graph $H$ satisfying the expansion condition required by
\Cref{lem:distributed-match}; once a matching saturating all jobs in $H$ is
found, the induced integral assignment is a $(2+\eps)$-approximation. Thus it
remains only to justify that this virtual graph and the matching algorithm can
be simulated in the \congest model on the original graph $G$.

\begin{claim}
\label{clm:constant-virtual-edges}
For every original edge $(i,j)$, at most three virtual edges of $H$ are
realized over $(i,j)$.
\end{claim}

\begin{proof}
Fix $(i,j)$. In the greedy grouping step of machine \(i\), the mass \(x_{ij}\) is
placed into consecutive virtual machines, each of capacity
\(\gamma=1/(1+\varepsilon)\). Since
$x_{ij}\le 1$ and $\eps\le 1/2$, this mass can intersect at most one
partially filled group at the beginning, at most one full group, and at most
one partially filled group at the end. Hence $j$ is adjacent to at most three
virtual machines hosted by $i$.
\end{proof}

\begin{proof}[Proof of \Cref{thm:rounding}]
We first show that \Cref{alg:rounding} completes in $\poly(\log(mn)/\eps)$ rounds.
By \Cref{clm:constant-virtual-edges}, every original edge $(i,j)$ realizes
only $O(1)$ virtual edges of $H$. Hence one communication round on $H$ can be
simulated by $O(1)$ rounds on $G$: each virtual edge message is sent over its
underlying edge $(i,j)$.

The construction of $H$ is local. Each machine $i$ forms its own virtual
machines, greedily places its incident fractional masses, and sends each job
$j$ the identifiers of the at most three virtual machines adjacent to it.
Thus building the locally represented virtual graph takes $O(1)$ \congest
rounds.

By the Li guarantee stated above, $H$ satisfies the expansion
condition of \Cref{lem:distributed-match}. Therefore
\texttt{DistributedMatch} finds a matching saturating all jobs in
$\poly(\log(mn)/\eps)$ rounds on $H$, and thus can be simulated in
\congest in $\poly(\log(mn)/\eps)$ rounds. Combining these,
\Cref{alg:rounding} completes in $\poly(\log(mn)/\eps)$ rounds.
Finally, \Cref{thm:rounding} directly follows from Li's rounding
guarantee that \Cref{alg:rounding} returns an integral assignment of
load at most $(2+\eps)\lambda^*$.
\end{proof}



\end{document}